\documentclass[letterpaper,journal,onecolumn,12pt]{IEEEtran}

\usepackage[utf8]{inputenc} 
\usepackage[T1]{fontenc}
\usepackage{verbatim}

\usepackage[backend=bibtex, style = numeric, doi = false, url = false, isbn = false, maxbibnames = 6]{biblatex}
\bibliography{references.bib}
\renewbibmacro{in:}{}      
\newbibmacro{string+doi}[1]{\iffieldundef{doi}{#1}{\href{https://dx.doi.org/\thefield{doi}}{#1}}}
\DeclareFieldFormat{title}{\usebibmacro{string+doi}{\mkbibemph{#1}}}
\DeclareFieldFormat[article]{title}{\usebibmacro{string+doi}{\mkbibquote{#1}}}
\DeclareFieldFormat[incollection]{title}{\usebibmacro{string+doi}{\mkbibquote{#1}}}                   
\DeclareFieldFormat[inproceedings]{title}{\usebibmacro{string+doi}{\mkbibquote{#1}}}     

\usepackage[colorlinks=true,linkcolor=blue,urlcolor=blue,citecolor=blue,anchorcolor=green,pdfusetitle]{hyperref}

\usepackage[cmex10]{amsmath}  
\usepackage{amsfonts}
\usepackage{amssymb,amsthm,mathtools,thmtools,booktabs,tabularx,float}

\usepackage{cleveref}

\usepackage{mleftright}       
\mleftright                   

\usepackage{graphicx}         
\usepackage{booktabs}         

\newtheorem{thm}{Theorem}
\newtheorem{prop}[thm]{Proposition}
\newtheorem{lem}[thm]{Lemma}
\newtheorem{cor}[thm]{Corollary}
\theoremstyle{definition}
\newtheorem{rem}[thm]{Remark}

\newtheorem{ex}[thm]{Example}

\usepackage{braket}
\usepackage{mathtools}
\usepackage{dsfont}
\usepackage{xcolor}

\newcommand{\cA}{\mathcal{A}}

\newcommand{\cD}{\mathcal{D}}
\newcommand{\cE}{\mathcal{E}}
\newcommand{\cF}{\mathcal{F}}
\newcommand{\cH}{\mathcal{H}}
\newcommand{\cN}{\mathcal{N}}

\newcommand{\bC}{\mathbb{C}}
\newcommand{\bE}{\mathbb{E}}
\newcommand{\bF}{\mathbb{F}}
\newcommand{\mc}{\mathcal}
\newcommand{\mbb}{\mathbb}

\newcommand{\tkgen}{\tilde{\kappa}_{\mathrm{gen}}}
\newcommand{\pwc}{p_{\text{wc}}}
\newcommand{\kwc}{k_{\text{wc}}}

\DeclareMathOperator{\End}{End}
\DeclareMathOperator{\mix}{mix}
\DeclareMathOperator{\Unif}{Unif}
\DeclareMathOperator{\law}{Law}
\DeclareMathOperator{\TV}{TV}
\DeclareMathOperator{\tr}{tr}
\DeclareMathOperator{\id}{id}

\newcommand{\one}{\mbb{I}}

\newcommand{\psucc}{p_{\mathrm{succ}}}

\newcommand{\kgen}{\kappa_{\mathrm{gen}}}

\newcommand{\fl}[1]{{\color{red}\noindent\textbf{Comment (FL):} #1}}

\definecolor{cool_green}{rgb}{0.0, 0.5, 0.0}

\usepackage[affil-it]{authblk}

\usepackage{amsthm}

\makeatletter
\let\originalrestatable\restatable
\renewcommand{\restatable}[3][]{%
  \originalrestatable[name={%
    #1%
    \ifthmt@thisistheone\else
      \if\relax\detokenize{#1}\relax\else; \fi
      restated from Section~\ref{sec:main-results}%
    \fi
  }]{#2}{#3}%
}
\makeatother

\title{A mixing time method for estimating the sample complexity of quantum state discrimination}
\author[1]{Juntai Zhou}
\author[1,2]{Felix Leditzky}
\affil[1]{Department of Mathematics, University of Illinois Urbana-Champaign}
\affil[2]{Illinois Quantum Information Science and Technology (IQUIST) Center, University of Illinois Urbana-Champaign}
\begin{document}
\maketitle

\begin{abstract}
    We develop a mixing time method for estimating the sample complexity of quantum state discrimination. We start with considering the minimum-error discrimination of geometrically uniform pure state ensembles, and prove that its sample complexity has a tight estimate given by a quantum homogeneous mixing time [George et al., 2026] and a quantum version of the generalized Dobrushin coefficient [Wolfer, 2020]. This quantum mixing time further reduces to a classical one when the generating group $G$ forms a Gelfand pair with the stabilizer subgroup $H$ of the generator state.
    In this case the generalized Dobrushin coefficient can be fully expressed by representation-theoretic quantities of the commutative Hecke algebra $\End_G(\bC[G/H])$. In particular, this method reduces the sample complexity estimation of learning quantum coupon collector states [Arunachalam et al., 2020] and learning phase states to classical mixing time problems. We apply this framework to answer the open problems of learning degree-$d$ phase states over $\bF_q$ in [Alrabiah et al., 2026] and generalized Boolean phase states over $\mathbb Z_q$ [Arunachalam et al., 2023]. The framework also applies to hypergraph state ensembles, giving estimates expressed fully in terms of hypergraph data and recovering estimates for graph state ensembles in [Montanaro and Shao, 2022]. Finally, we extend the discussion to arbitrary mixed state ensembles with uniform priors, prove a sandwiched bound for minimum-error discrimination sample complexity by a quantum weakly mixing time, and provide a tight estimate for the minimax discrimination sample complexity from [D’Ariano et al., 2005] by a Dobrushin-type coefficient. We also discuss the method of strengthened data processing inequality [Gao and Rouz{\'e}, 2022] and give an upper bound in terms of a strengthened data processing inequality constant.
\end{abstract}

\section{Introduction}
The problem of minimum-error quantum state discrimination seeks a measurement that 
identifies with high probability one of $n$ quantum states, each prepared with some prior probability \cite{bae2015quantum,barnett2009quantum}. Its sample complexity is the least number of i.i.d.~copies required to achieve success probability of at least $1-\epsilon$,
\begin{align}
    k_{\min}(\epsilon)\coloneqq \min\{k:\psucc^*(k)\geq1-\epsilon\},
\end{align}
where $\psucc^*(k)$ is the optimal success probability of discriminating the $k$-copy ensemble. This formulation occurs throughout quantum information theory, quantum learning theory and quantum algorithms: For example, 
the fidelity of teleportation protocols is proportional to the success probability of an associated quantum state discrimination problem \cite{chitambar2024teleportation,ishizaka2008asymptotic};
optimal
adaptive discrimination of a finite family of jointly
teleportation-covariant channels is equivalent to
discriminating tensor powers of the corresponding Choi states \cite{ZhuangPirandola2020}; the hidden subgroup and generalized hidden shift problems lead to discrimination tasks for group-structured quantum states \cite{BaconChildsVanDam2006, childs2007quantum, zhou2026sample}; learning phase states and identifying graph states likewise require distinguishing large, highly structured families \cite{Arunachalam2020QuantumCouponCollector,arunachalam_et_al:LIPIcs.TQC.2023.3,MontanaroShao2022HiddenGraph}; and discrimination between finite collections of separated states also underlies information-theoretic lower bounds for tomography~\cite{HaahEtAl2017}. Estimating $k_{\min}$ therefore characterizes the sample complexity (or query complexity) of these problems.

The Barnum--Knill bound \cite{BarnumKnill2002} gives a general upper bound for the sample complexity $k_{\min}$ in terms of pairwise fidelities; in particular, it implies an unconditional $O(\log n)$-copy upper bound at fixed target error provided that the largest pairwise fidelity is uniformly bounded away from one~\cite{harrow2012copies,Montanaro2019PrettySimpleBounds}. Such bounds are useful and broadly applicable, but may be far from the true complexity because they can fail to capture the collective geometry that makes a large ensemble distinguishable with substantially fewer copies. Obtaining tight estimates for structured families often still requires a problem-specific analysis of the measurement~\cite{alrabiah2026nolowdegree,Arunachalam2020QuantumCouponCollector,arunachalam_et_al:LIPIcs.TQC.2023.3,BaconChildsVanDam2006,childs2007quantum,MontanaroShao2022HiddenGraph,zhou2026sample}.

The difficulty in systematically estimating the sample-complexity scaling is not due to insufficient information about discrimination. For any pure state ensemble, the optimal success probability of discrimination can be written as an SDP given the Gram matrix \cite[eq.~(21)]{mohan2026hybrid}. Since the Gram matrix of $k$-copy states is simply $X^{\circ k}$ where $\circ$ denotes the Hadamard product, $X$ contains sufficient information to determine the sample complexity of discrimination. What remains is to organize its dependence on the number of copies $k$ into a form that can be estimated without solving each growing measurement problem separately. Developing such a systematic route beyond coarse pairwise-overlap estimates and case-by-case analysis is the central motivation and main contribution of this article.

\section{Overview of main results}
\label{sec:main-results}

In this article, we estimate the sample complexity of quantum state discrimination by relating it to a mixing time problem. The central results are \Cref{theorem: GU Dobrushin estimate}, which gives a tight estimate for the discrimination sample
complexity in terms of a quantum mixing time and a Dobrushin-type coefficient, and \Cref{theorem: classical random walk on dual hypergroup}, which reduces the corresponding quantum mixing time to a classical one. The rest of the main results are extensions of \Cref{theorem: GU Dobrushin estimate} and applications of \Cref{theorem: classical random walk on dual hypergroup}.

The following subsections present a more detailed overview of our main results. Full proofs can be found in the subsequent \Cref{section: GU ensemble,section: phase state,section: generic non-GU ensembles}.
We also provide a summary of these results in \Cref{tab:summary-of-results}.

\subsection{From Gram matrix to mixing times}
\label{sec:gram-to-mixing-time}

The mixing time of a Markov random walk is the time it takes for the walk’s probability distribution to become close to its stationary distribution \cite{LevinPeres2017MarkovChains}. Its quantum analogue replaces probability distributions by density operators and Markov kernels by quantum channels~\cite{George2026QuantumDoeblinCoefficients, HiaiRuskai2016}. Our starting observation is that taking additional copies of geometrically uniform (GU) pure states has a simple mixing-time interpretation: For these ensembles, the `pretty good measurement'' (PGM; defined in \eqref{eq:pretty-good-measurement} below) \cite{belavkin1975optimal,hausladen1994pretty,holevo1979asymptotically} is optimal  with $k$-copy success probability \cite{Eldar2004OptimalDetection,zhou2025distinguishability}
\begin{equation}
  p_{\mathrm{succ}}^*(k)
  =\left(\frac{1}{|G|}\tr\sqrt{X_k}\right)^2=F\!\left(\Phi_X^{\,k}(\tau),\omega\right),
\end{equation}
where $F(\rho,\sigma)$ is the squared fidelity, $\Phi_X$ is the Hadamard channel of Gram matrix $X$, $\tau=\frac{1}{|G|}\mathbb J$ is the normalized state of the all-one matrix $\mathbb J$, and $\omega=\frac{1}{|G|}\one$ is the completely mixed state on $\bC^{|G|}$. The sample complexity therefore coincides with the fidelity-distance mixing time of the channel $\Phi_X$ for initial state $\tau$ and stationary state $\omega$. Standard comparisons between fidelity and trace distance then give two-sided bounds in terms of the trace-norm mixing time \begin{align}
    t_{\min}(\epsilon)\coloneqq\min\{t:\|\Phi_X^t(\tau)-\omega\|_1\leq\epsilon\}.
\end{align}
This can be further extended to arbitrary uniform mixed state ensemble $\{\rho_i=S_iS_i^\dagger\}_{i=1}^{n}$. We define
\begin{align}
    t_{\min}(\epsilon)\coloneqq \min\{t:\tfrac{1}{n}\|X_t-\Delta_t(X_t)\|_1\leq\epsilon\}
\end{align}
where $X_t\coloneqq \begin{bmatrix}(S_i^\dagger S_j)^{\otimes t}\end{bmatrix}_{i,j=1}^{n}$ and $\Delta_t(X_t)$ is the block pinching of $X_t$, and we prove the following sandwiched bound for the sample complexity:
\begin{restatable}{thm}{MixedStateSandwich}\label{theorem: sandwiched bound for sample complexity}
    $t_{\min}$ gives a sandwiched bound for the sample complexity of state discrimination:
    \begin{align} \label{eq: sandwiched bound for sample complexity}
        t_{\min}(4\sqrt{\epsilon})\leq k_{\min}(\epsilon)\leq t_{\min}(\epsilon).
    \end{align}
\end{restatable}

Let $H\coloneqq \{h\in G:U_h\ket{\psi}=\ket{\psi}\}$ be the stabilizer subgroup of the generator state $\ket{\psi}$, then the states are indistinguishable within the same coset $gH$, so the meaningful question reduces to discriminating the states $\ket{\psi_{gH}}$ labeled by $G/H$. A technical issue is that the Hadamard channel fixes every diagonal matrix, so it does not mix to $\omega$ from every state on the full matrix algebra; to resolve this obstruction, we observe that the Gram matrix $X$ lives in the Hecke algebra $\cA=\End_G(\bC[G/H])$ (see \Cref{sec:main-results-dual-hypergroup}) and restrict its Hadamard channel to $\cA$. Within this algebra, we show that the fixed initial state $\tau$ is a worst-case starting state for convergence to $\omega$. This identifies the sample complexity with the worst-case mixing time 
\begin{align}
    t_{\mix}(T,\omega,\epsilon)\coloneqq\min \left\{t:\sup\nolimits_{\rho\in\cD(\cA)}\|T^t(\rho)-\omega\|_1\leq\epsilon \right\}
\end{align}
to which contraction coefficients naturally apply.

\subsection{Contraction coefficients and the mixing scale}
\label{sec:contraction-coefficients}

The classical Dobrushin coefficient $\kappa(P)$ of Markov kernel $P$ is the maximal fraction of total variation distance preserved after one transition \cite{dobrushin1956central1,dobrushin1956central2}. For a stationary distribution $\pi$, it gives the contraction property $d_{\TV}(\mu P^t,\pi)\leq\kappa(P)^td_{\TV}(\mu,\pi)$, ensuring geometric convergence to stationarity when $\kappa(P)<1$. However, an ergodic chain may have $\kappa(P)=1$, in which case this single-step coefficient fails to capture contraction despite eventual convergence at later times. Wolfer defines the generalized Dobrushin coefficient \cite{wolfer2020mixing}
\begin{equation}
  \kgen(P)
  =1-\sup_{s\ge1}\frac{1-\kappa(P^s)}{s}
  \label{eq:generalized-coefficient}
\end{equation}
where $1-\kgen(P)$ captures the best per-step contraction over multistep blocks, and proves that its reciprocal characterizes the mixing time up to constants at fixed accuracy.

For a quantum channel $T$, the analogous Dobrushin coefficient is defined in \cite{HiaiRuskai2016}. In this article, we extend Wolfer's definition to quantum channels by replacing Markov kernel $P$ with quantum channel $T$ in~\eqref{eq:generalized-coefficient}, and show that it characterizes the asymptotic scale of the mixing time:
\begin{restatable}{thm}{DobrushinTight} \label{theorem: quantum Dobrushin estimate}
    Let $\epsilon\in(0,\frac{1}{2})$. If $\kgen(T)<1$, then
    \begin{align}
        \frac{1-\epsilon}{1-\kgen(T)}\leq t_{\mix}(T,\omega,\epsilon)\leq\frac{1+\log(\frac{2}{\epsilon})}{1-\kgen(T)}.
    \end{align}
\end{restatable}
Our main result applies this bound to $T=\Phi_X\big|_\cA$ and gives a tight estimate for the discrimination sample complexity in terms of $t_{\min}$ and the Dobrushin-type coefficient:

\begin{restatable}{thm}{GUDobrushin}
\label{theorem: GU Dobrushin estimate}
    Given a GU ensemble generated by $G$ and $\ket{\psi}$, let $H$ be the stabilizer subgroup of $\ket{\psi}$. If $\max_{g\notin H}|\bra{\psi}U_g\ket{\psi}|<1$, then the quantum mixing time $t_{\min}$ and the generalized quantum Dobrushin coefficient $\kgen$ both give a tight estimate for the discrimination sample complexity of the states $\ket{\psi_{gH}}$:
    \begin{align}
        k_{\min}=\Theta_\epsilon(t_{\min})=\Theta_\epsilon\bigg(\frac{1}{1-\kgen(\Phi_X\big|_\cA)}\bigg)
    \end{align}
    for $\epsilon\in(0,\frac{1}{2})$ where $\cA=\End_G(\bC[G/H])$ is the Hecke algebra.
\end{restatable}

The discussion can be further extended to arbitrary mixed state ensembles without symmetry assumptions, where unconditional tight estimate can be obtained for the worst-case (a.k.a.~minimax) discrimination \cite{DAriano2005Minimax,Montanaro2019PrettySimpleBounds} sample complexity. We first prove in \Cref{proposition: worst-case sample complexity equivalence} that this sample complexity is stable under constant changes in the error threshold.
This enables the following tight estimate by the Dobrushin-type coefficient:
\begin{restatable}{thm}{WorstCaseDiscrimination} \label{theorem: tightness for worst-case discrimination}
    Let $\kwc(\epsilon)$ denote the worst-case discrimination sample complexity, then
    \begin{align}
        \kwc=\Theta_\epsilon\left(\frac{1}{1-\tkgen}\right)\quad\forall\:\epsilon\in(0,\frac{1}{2}).
    \end{align}
    In particular,
    \begin{align} \label{eq: relations for general ensembles}
        t_{\min}(4\sqrt{\epsilon})\leq k_{\min}(\epsilon)\leq t_{\min}(\epsilon)\leq O_\epsilon(\kwc)=\Theta_\epsilon\left(\frac{1}{1-\tkgen}\right).
    \end{align}
\end{restatable}


\subsection{Classical mixing time on the dual hypergroup}
\label{sec:main-results-dual-hypergroup}
For any GU ensemble, the Hecke algebra $\End_G(\bC[G/H])$ is commutative if and only if $(G,H)$ is a Gelfand pair \cite{CeccheriniSilbersteinScarabottiTolli2008, ScarabottiTolli2013}. In this case, the commutative Hecke algebra naturally induces a dual hypergroup $\hat\Omega$ with convolution determined by the normalized Krein parameters \cite{Godsil2010AssociationSchemes,Voit2017GeneralizedCommutativeAssociationSchemes}, and any probability measure $p$ on $\hat{\Omega}$ defines a Markov kernel $K_p$ by convolution \cite{Voit2019ContinuousAssociationSchemes}. We show that any correlation matrix $Y$ in the Hecke algebra bijectively maps to a probability measure $p(Y)$ on $\hat{\Omega}$, and the trace distance between $\Phi_Y^t(\tau_n)$ and $\omega_n$ is equal to the total variation distance between the $t$-th convolution power of $p(Y)$ and the Haar measure of $\hat{\Omega}$:

\begin{restatable}{thm}{ClassicalMixing} \label{theorem: classical random walk on dual hypergroup}
    There is a bijection 
    \begin{align} \label{eq: the map p}
        p\colon\{Y\in\End_G(\bC[G/H]):Y\geq0,\tr(Y)=n\}&\to\mc P(\hat\Omega),
    \end{align}
    where $\mc P(\hat\Omega)$ denotes the space of probability measures on $\hat\Omega$, such that 
    \begin{align}
        \|\Phi_X^t(\tau_n)-\omega_n\|_1=2d_{\TV}(p(X)^{*t},\pi)
    \end{align}
    where $X$ is the Gram matrix. Therefore, the quantum mixing time $t_{\min}$ reduces to a classical total-variation mixing time, where the corresponding classical Markov kernel is given by
    \begin{align}
        K(i,j)\coloneqq(\delta_i*p(X))(j).
    \end{align}
\end{restatable}

This theorem reduces the sample complexity estimate to the mixing of the classical walk $K$ on the dual hypergroup. The quantum Dobrushin coefficient of $\Phi_X\big|_\cA$ also reduces to the classical Dobrushin coefficient of $K$, and we obtain an explicit formula \cref{eq: Gelfand Dobrushin formula} for the tight estimate of sample complexity completely in terms of representation-theoretic information including Krein parameters and primitive idempotents.

The Gelfand-pair framework brings several established and new examples into a unified form. The first example is learning quantum coupon collector states~\cite{Arunachalam2020QuantumCouponCollector} where permutation symmetry yields a Gelfand pair and the dual walk is a birth-death chain. Applying our framework to this problem recovers the known sample-complexity scaling by applying classical mixing time theory.

We then consider a class of GU phase-state ensembles generated by an abelian group $A$ with
    \begin{align} \label{eq: GU phase state with feature map}
        \ket{\psi_0}=\sum_{x\in\mc X}\sqrt{\nu(x)}\ket{x},\qquad U_a\ket{x}=\Phi(x)(a)\ket{x} \text{ for $a\in A,x\in\mc X,$}
    \end{align}
    where $\mc X$ is a finite set with probability measure $\nu$, and $\Phi\colon\mc X\to\hat A$ is a feature map from $\mc X$ into the character group $\hat{A}$ of $A$. Set $p=\Phi_\#\nu$ to be the pushforward measure of $\nu$ under $\Phi$. Our next result provides an explicit classical random walk $Z_t$ on $\hat A$ whose mixing time is equal to the $t_{\min}$ of the phase state ensemble: 

\begin{restatable}{thm}{PhaseState}
    \label{theorem: semidirect product case reduction to d_TV}
    Let $(Z_t)$ be the abelian random walk on $\hat{A}$ given by
    \begin{align}
        Z_0=0;\quad Z_t=\sum_{j=1}^{t}X_j
    \end{align}
    where $X_j\overset{\text{i.i.d.}}{\sim}p$. Then
    \begin{align}
        t_{\min}(\epsilon)=\min\{t:\|\law(Z_t)-\Unif(\hat{A})\|_1\leq\epsilon\}.
    \end{align}
\end{restatable}
This perspective is compatible with the recent characterization of codeword-state learnability through classical syndrome distributions~\cite{alrabiah2026nolowdegree} and supplies a general framework to this class of phase states. The generalized Dobrushin coefficient of this abelian walk can be fully expressed in terms of $A,\nu,\Phi$. When the dual group admits a product structure $\hat{A}=B_1\times\ldots\times B_n$ and the feature map has a Knothe--Rosenblatt triangular structure~\cite{ramgraber2025friendly}, we can upper bound mixing time by information of the coordinate blocks rather than a single bound on the full family of features:

\begin{restatable}{thm}{Triangular}
    \label{thm: total variation upper bound for triangular feature map}
        Assume $X^1,\ldots,X^n$ are independent. Define
    \begin{align}
        \lambda_i\coloneqq\max_{\chi\in\hat{B}_i\backslash\{1\}}\bE_{X^{<i}}|\bE_{X^i}\chi(\Phi_i(X^{\leq i}))|^2.
    \end{align}
    Then for every $t>1$,
    \begin{align}
        d_{\TV}(\law(Z_t),\Unif(\hat{A}))\leq\frac{1}{2}\sum_{i=1}^{n}\sqrt{(|B_i|-1)\lambda_i^t}.
    \end{align}
    In particular, setting $\lambda=\max_i\lambda_i$ and $r=\max_i\log|B_i|$, 
    \begin{align} \label{eq: upper bound for triangular feature map}
        t_{\min}(\epsilon)\leq\frac{r+2\log(\frac{n}{\epsilon})}{-\log\lambda}.
    \end{align}
\end{restatable}

As applications, this framework gives an $O_{q,d,\epsilon}(n^{d-1})$ copy upper bound for learning degree-$d$ phase states on $\bF_q^n$ for fixed prime $q$ and fixed $d\ge2$, which closes the upper-bound conjecture posed by Alrabiah, Arunachalam, Grewal, and Wright~\cite{alrabiah2026nolowdegree}. The same method gives the same upper bound for generalized degree-$d$ Boolean phase states with phases in $\mathbb Z_q$ for fixed even $q$, answering the collective-measurement question left open by Arunachalam, Bravyi, Dutt, and Yoder~\cite{arunachalam_et_al:LIPIcs.TQC.2023.3}. Finally, hypergraph-state ensembles~\cite{Rossi2013QuantumHypergraphStates} fit the feature-map framework and the generalized Dobrushin coefficient can be entirely expressed in terms of hypergraph data in \cref{eq: hypergraph Dobrushin coefficient}. We also derive the following upper bound in terms of hyperedge information in \Cref{proposition: hypergraph}, which recovers results of Montanaro and Shao~\cite{MontanaroShao2022HiddenGraph} when applied to graph states.

\subsection{Summary of results in a table}
\begin{table}[H]
\centering
\small
\setlength{\tabcolsep}{4pt}
\renewcommand{\arraystretch}{1.2}
\begin{tabularx}{\textwidth}{@{}
  >{\raggedright\arraybackslash}p{0.22\textwidth}
  >{\raggedright\arraybackslash}p{0.25\textwidth}
  >{\raggedright\arraybackslash}X@{}}
\toprule
\textbf{Ensemble type} & \textbf{Mixing formulation} & \textbf{Main results} \\
\midrule
Arbitrary mixed states (with uniform priors)
& Inhomogeneous quantum mixing
& $t_{\min}(4\sqrt{\epsilon})\le k_{\min}(\epsilon)\le t_{\min}(\epsilon)\leq O_\epsilon(\kwc)=\Theta_\epsilon(\frac{1}{1-\tkgen})$ in \Cref{theorem: tightness for worst-case discrimination}.
\\ \addlinespace

Arbitrary pure states $\quad$ (with uniform priors)
& Homogeneous Hadamard-channel mixing
& Same bounds with unstabilized coefficients; strengthened data processing inequality (SDPI) bound \eqref{eq: SDPI}. 
\\ \addlinespace

GU pure states
& Hadamard-channel mixing on the Hecke algebra $\mathcal A$
& $k_{\min}=\Theta_\epsilon(t_{\min})=\Theta_\epsilon\!\left(
    (1-\kappa_{\mathrm{gen}}(\Phi_X|_{\mathcal A}))^{-1}
  \right)$ in \Cref{theorem: GU Dobrushin estimate}.
\\ \addlinespace

Gelfand-pair GU pure states
& Classical walk $K$ on the dual hypergroup
& $k_{\min}=\Theta_\epsilon(t_{\mix}(K))=\Theta_\epsilon((1-\kappa_{\mathrm{gen}}(K))^{-1})$;
  explicit formula \eqref{eq: Gelfand Dobrushin formula} for $\kgen$ from representation-theoretic data.
\\ \addlinespace

Feature-map phase states
& Abelian walk $Z_t=\sum_{j=1}^t\Phi(X_j)$
& Formula \eqref{eq:feature-lifted-walk-expanded} for the abelian walk; explicit formula \eqref{eq: phase state Dobrushin} for $\kappa_{\mathrm{gen}}$;
  triangular-feature mixing bound \eqref{eq: upper bound for triangular feature map}.
\\ \midrule

Quantum coupon collector
& Birth--death chain from the Gelfand pair
  $(S_\ell,S_r\times S_{\ell-r})$
& $k_{\min}=\Theta_\epsilon\!\left(
    r\log\min\{r,\ell-r\}
  \right)$ in \cref{eq: coupon collector estimate}.
\\ \addlinespace

Degree-$d$ phase states over $\mathbb F_q^n$
& Triangular monomial features
& $k_{\min}=\Theta_{q,d,\epsilon}(n^{d-1})$ for prime $q$ in \cref{eq: degree-d phase state estimate}.
\\ \addlinespace

Generalized degree-$d$ Boolean phase states
& Triangular monomial features
& $k_{\min}=\Theta_{q,d,\epsilon}(n^{d-1})$ for even $q$ in \cref{eq: Boolean phase state estimate}.
\\ \addlinespace

Hypergraph states
& Triangular-feature walk on $\mathbb F_2^E$
& Explicit formula \eqref{eq: hypergraph Dobrushin coefficient} for $\kgen$ from hyperedge data;
$k_{\min}(\epsilon)\le2^d\log\!\left(\epsilon^{-1}\sum_{v\in V}2^{|E_v|}\right)$ in \Cref{proposition: hypergraph}.
\\ \bottomrule
\end{tabularx}
\caption{Summary of our own main results}
\label{tab:summary-of-results}
\end{table}

\subsection{Organization of this article}
The rest of this article is organized as follows: \Cref{section: preliminaries} collects the necessary background on discrimination, mixing, and contraction coefficients. \Cref{section: GU ensemble} develops the GU pure-state theory and its classical reduction. \Cref{section: phase state} treats phase-state and hypergraph-state applications. \Cref{section: generic non-GU ensembles} extends the framework to arbitrary uniform mixed-state ensembles and establishes the worst-case and global-geometry characterizations.

\section{Preliminaries} \label{section: preliminaries}
\subsection{Minimum-error state discrimination and its sample complexity}
    Given an ensemble $\{(p_i,\rho_i)\}_{i=1}^{n}$ of quantum states $\rho_i$ each drawn with probability $p_i$, we define the success probability of minimum-error (a.k.a. average-error) state discrimination by
    \begin{align}
        \psucc^*\coloneqq \max_{\{\Pi_i\}}\sum_i p_i\tr(\Pi_i\rho_i).
    \end{align}
    In particular, consider the success probability of discriminating $k$-copy i.i.d. states
    \begin{align}
        \psucc^*(k)\coloneqq \max_{\{\Pi_i^{(k)}\}}\sum_i p_i\tr(\Pi_i^{(k)}\rho_i^{\otimes k}).
    \end{align}
    Note that $\lim_{k\to\infty}\psucc(k)=1$ provided that the $\rho_i$ are pairwise different. Given $\epsilon\in(0,1)$, define
    \begin{align}
        k_{\min}(\epsilon)\coloneqq \min\{k:\psucc^*(k)\geq1-\epsilon\}.
    \end{align}
    In this work we focus on studying the asymptotic order of $k_{\min}(\epsilon)$ given a family of ensembles. 

The `pretty good measurement' (PGM) \cite{belavkin1975optimal,holevo1979asymptotically,hausladen1994pretty} is a particular measurement $\Pi^{\mathrm{PGM}} = \lbrace \Pi^{\mathrm{PGM}}_i\rbrace_{i=1}^n$ with measurement operators defined in terms of the quantum state ensemble $\{(p_i,\rho_i)\}_{i=1}^{n}$ as
\begin{align}
    \Pi^{\mathrm{PGM}}_i \coloneqq \bar{\rho}^{-1/2} p_i \rho_i \bar{\rho}^{-1/2},
    \label{eq:pretty-good-measurement}
\end{align}
where $\bar{\rho} = \sum_{i=1}^n p_i \rho_i$ is the ensemble average state.
The Barnum-Knill bound \cite{BarnumKnill2002} provides an upper bound for the error probability of PGM by pairwise fidelities, which leads to a simple but generic upper bound for the sample complexity:

\begin{lem}[Barnum-Knill]\label{lemma: Barnum-Knill}
    Given an ensemble $\mathcal{E}\coloneqq \{(p_i,\rho_i)\}$,
    \begin{align}
        1-p_{\mathrm{PGM}}\leq\sum_{i\neq j}\sqrt{p_ip_j}\sqrt{F}(\rho_i,\rho_j).
    \end{align}
\end{lem}

\begin{prop}[Montanaro \cite{Montanaro2019PrettySimpleBounds}]\label{theorem: Montanaro upper bound of k_min}
    Suppose $p_i=\frac{1}{n}$. If $\mu(\mathcal{E})\coloneqq \max_{i\neq j}F(\rho_i,\rho_j)<1$, then
    \begin{align}
        k_{\min}(\epsilon)\leq O\left(\frac{\log\epsilon-\log n}{\log\mu}\right).
    \end{align}
\end{prop}
\begin{proof}
    Since fidelity is multiplicative, by Barnum-Knill,
    \begin{align}
        1-p_k\leq\frac{1}{n}\sum_{i\neq j}\sqrt{F}(\rho_i,\rho_j)^k\leq\frac{1}{n}\sum_{i\neq j}\mu^\frac{k}{2}=(n-1)\mu^\frac{k}{2}.
    \end{align}
    Therefore, $n\mu^\frac{k}{2}\leq\epsilon$ suffices to ensure $p_k\geq1-\epsilon$, so we have $k_{\min}\leq O(\frac{\log\epsilon-\log n}{\log\mu})$.
\end{proof}

\begin{rem}
    \Cref{theorem: Montanaro upper bound of k_min} holds for general state ensembles, but it only gives an $O_{\epsilon,\mu}(\log n)$ upper bound. There are many examples where one can show tighter upper bounds on $k_{\min}$, and even $k_{\min}=O(1)$ \cite{alrabiah2026nolowdegree,Arunachalam2020QuantumCouponCollector,arunachalam_et_al:LIPIcs.TQC.2023.3,BaconChildsVanDam2006,childs2007quantum,MontanaroShao2022HiddenGraph,zhou2026sample}; however, a systematic method to upper bound $k_{\min}$ better than $O(\log n)$ is an open question \cite{Cheng2025InvitationSampleComplexityQHT}. Our work answers this question by providing a mixing time method.
\end{rem}

\subsection{Worst-case state discrimination and its sample complexity}
Another setup of state discrimination considers the worst-case probability of failure \cite{DAriano2005Minimax}. Define 
\begin{align}
        \pwc(t)\coloneqq\max_{M_i^{(t)}}\min_{i}\tr\left(M_i^{(t)}\rho_i^{\otimes t}\right)
\end{align}
to be the worst-case discrimination success probability for $t$-copy ensemble, and
\begin{align}
    \kwc(\epsilon)\coloneqq\min\{t:\pwc(t)\geq1-\epsilon\}
\end{align}
to be the worst-case discrimination sample complexity. We prove that this sample complexity is stable under constant changes in the error threshold $\epsilon$:
\begin{prop} \label{proposition: worst-case sample complexity equivalence}
    For any $\epsilon\leq\delta\in(0,\frac{1}{2})$, we have
    \begin{align}
        \kwc(\epsilon)=\Theta_{\epsilon,\delta}\kwc(\delta).
    \end{align}
\end{prop}
\begin{proof}
    Let $M$ be an optimal $\kwc(\delta)-$copy measurement for worst-case discrimination, then applying it to $m$ independent groups of $\kwc(\delta)-$copy states gives outcome $Z=(Z_1,\ldots,Z_m)$ with $P_i(Z_r=i)\geq 1-\delta$ where $P_i$ denotes the conditional probability given the true state is $\rho_i$. Applying \cite[Proposition 2.2]{kunsch2019solvable} with interval $I=\{i\}$ for each $i$ gives
    \begin{align}
        \min_iP_i[\text{med}(Z)=i]\geq1-\frac{1}{2}(4\delta(1-\delta))^\frac{m}{2}.
    \end{align}
    Define
    \begin{align}
        \Pi_j\coloneqq\sum_{\substack{z\in[n]^m\\\text{med}(z)=j}}M_{z_1}\otimes\ldots\otimes M_{z_m},
    \end{align}
    which is positive and satisfies $\sum_j\Pi_j=\sum_zM_{z_1}\otimes\ldots\otimes M_{z_m}=(\sum_iM_i)^{\otimes m}=\one$, thus a POVM. Moreover,
    \begin{align}
        \tr(\Pi_j\rho_i^{\otimes m\kwc(\delta)})=\sum_{\text{med}(z)=j}\prod_r\tr(M_{z_r}\rho_i^{\otimes\kwc(\delta)})=\sum_{\text{med}(z)=j}P_i(Z=z)=P_i(\text{med}(Z)=j).
    \end{align}
    Choosing odd integer $m=m(\epsilon,\delta)$ such that $\frac{1}{2}(4\delta(1-\delta))^\frac{m}{2}\leq\epsilon$, we obtain
    \begin{align}
        \pwc(m\kwc(\delta))\geq\min_i\tr(\Pi_i\rho_i^{\otimes m\kwc(\delta)})=\min_iP_i(\text{med}(Z)=i)\geq1-\epsilon,
    \end{align}
    therefore $\kwc(\delta)\leq\kwc(\epsilon)\leq m(\epsilon,\delta)\kwc(\delta)$.
\end{proof}

Note that this equivalence of sample complexity does not hold in general for minimum-error state discrimination. A simple example is given by 
\begin{align} \label{eq: min-error simple counterexample}
    \rho_1=\ket{0}\bra{0},\quad\rho_2=\ket{1}\bra{1},\quad\rho_3=2^{-\frac{1}{n}}\ket{1}\bra{1}+(1-2^{-\frac{1}{n}})\ket{2}\bra{2}
\end{align}
where the optimal minimum-error discrimination success probability is $\psucc^*(t)=1-\frac{2^{-\frac{t}{n}}}{3}$, therefore $k_{\min}(\frac{1}{3})=1$ while $k_{\min}(\frac{1}{6})=n$. The reason is that $\rho_3$ is almost identical with $\rho_2$ and difficult to distinguish, but error threshold above $\frac{1}{3}$ permits sacrificing the third state entirely, which is not allowed in worst-case discrimination.

\subsection{Basics of classical Markov chain mixing}
For probability measures $\mu,\nu$ on a finite set $\Omega$, define \cite{LevinPeres2017MarkovChains}
\begin{align}
    d_{\TV}(\mu,\nu)
  \coloneqq\sup_{B\subseteq\Omega}|\mu(B)-\nu(B)|=\frac12\sum_{x\in\Omega}|\mu(x)-\nu(x)|
  .
\end{align}
If a Markov chain $K$ on $\Omega$ has stationary law $\pi$, define the (worst-case initial state) mixing time by
\begin{align}
    t_{\mix}(K,\epsilon)\coloneqq \min\{t:\sup_{x_0}d_{\TV}(\delta_{x_0}K^t,\pi)\le\epsilon\}
\end{align}
where $\mu K(z)\coloneqq\sum_{x\in\Omega}\mu(x)K(x,z)$.

A simple but important class of Markov process is called birth-death process, defined as a Markov chain on an ordered state space (usually $\mathbb Z_{\geq0}$) in which transitions are permitted only between neighboring states. The following estimate of mixing time for birth-death chains is \cite[Proposition 4.1]{chen2013comparison}:

\begin{prop} \label{prop: birth-death mixing}
    Let $K$ be a discrete-time irreducible birth-death chain on $\{0,\ldots,d\}$ with nonnegative eigenvalues $1,\lambda_1,\ldots,\lambda_d$. Then for any $\epsilon\in(0,\frac{1}{8})$,
    \begin{align}
        t_{\mix}(K,\epsilon)=\Theta_\epsilon(\sum_{j=1}^{d}\frac{1}{1-\lambda_j}).
    \end{align}
\end{prop}

We will also use the following inequality:
\begin{lem} \label{lemma: total variation chain inequality}
    Let $Y=(Y_1,\ldots,Y_s)$ be a random vector whose entries $Y_i$ take values in $B_i$.
    Then,
    \begin{align}
        d_{\TV}(\law(Y),\Unif(B_1\times\ldots\times B_s))\leq\sum_{i=1}^{s}\bE d_{\TV}(\law(Y_i|Y_{<i}),\Unif(B_i)).
    \end{align}
    In particular, if $Y_{<i}$ is measurable with respect to a sigma-algebra $\cF_{i-1}$, then the $i$-th summand is at most $\bE d_{\TV}(\law(Y_i|\cF_{i-1}),\Unif(B_i))$.
\end{lem}
\begin{proof}
    This directly follows from \cite[Fact A.6]{assadi2025rounds}.
\end{proof}

\subsection{Dobrushin coefficients}
\subsubsection{Classical Dobrushin coefficients}
Given a Markov kernel $K$ on $\Omega$, the Dobrushin coefficient $\kappa(K)$ is defined as \cite{bremaud1999markov,dobrushin1956central1, dobrushin1956central2}
\begin{align}
    \kappa(K)\coloneqq\max_{(i,j)\in\Omega^2}d_{\TV}(K(i,\cdot),K(j,\cdot)).
\end{align}
It characterizes the contraction rate of $P$ because
\begin{align}
    \kappa(K)=\sup_{\mu\neq\nu}\frac{d_{\TV}(\mu K,\nu K)}{d_{\TV}(\mu,\nu)}.
\end{align}
It ensures geometric convergence to a stationary distribution $\pi$ when $\kappa(K)<1$ because $d_{\TV}(\mu K^t,\pi)\leq\kappa(K)^td_{\TV}(\mu,\pi)$. However, an ergodic chain may have $\kappa(K)=1$, in which case this single-step coefficient fails to capture contraction despite eventual convergence at later times. Wolfer defines the generalized Dobrushin contraction coefficient as \cite{wolfer2020mixing}
\begin{align}
    \kgen(K)\coloneqq1-\max_{s\in\mathbb{N}}\frac{1-\kappa(K^s)}{s},
    \label{eq:wolfer-generalized-dobrushin}
\end{align}
and proves that it characterizes the mixing time up to universal constants:

\begin{thm} \cite[Theorem 1]{wolfer2020mixing}
    Let $\epsilon\in (0, 1/2)$, and $K$ ergodic with mixing time $t_{\mix}(K,\epsilon)$. Then
    \begin{align}
        \frac{1-2\epsilon}{1-\kgen(K)}\leq t_{\mix}(K,\epsilon)\leq\frac{1-\log\epsilon}{1-\kgen(K)}.
    \end{align}
\end{thm}

\subsubsection{Quantum Dobrushin coefficient} \label{subsubsection: Quantum Dobrushin coefficient}
Let $\cA$ be a finite-dimensional unital $C^*$-subalgebra and $T\colon \cA\to\cA$ a quantum channel. \textcite{HiaiRuskai2016} define the quantum Dobrushin coefficient as
\begin{align} \label{eq: defn of quantum Dobrushin coefficient}
    \kappa(T)\coloneqq\sup_{\rho\neq\sigma\in\cD(\cA)}\frac{\|T(\rho)-T(\sigma)\|_1}{\|\rho-\sigma\|_1}
\end{align}
with equivalent forms
\begin{align}
    \kappa(T)=\frac{1}{2}\sup_{\rho,\sigma\in\cD(\cA)}\|T(\rho)-T(\sigma)\|_1=\sup_{\substack{B=B^\dagger\in\cA\\\tr(B)=0\\B\neq 0}}\frac{\|T(B)\|_1}{\|B\|_1}.
\end{align}
In this article, we generalize Wolfer's definition \eqref{eq:wolfer-generalized-dobrushin} of the generalized Dobrushin coefficient to the quantum setting:
\begin{align}
    \kgen(T)\coloneqq1-\sup_{s\in\mathbb{N}}\frac{1-\kappa(T^s)}{s}.
\end{align}
Given a stationary state $\omega$ such that $T(\omega)=\omega$, we define the (worst-case initial-state) quantum homogeneous mixing time \cite{George2026QuantumDoeblinCoefficients} of $T$ by
\begin{align} \label{eq: worst case mixing time for generic channel}
    t_{\mix}(T,\omega,\epsilon)\coloneqq\min\{t:\sup_{\rho\in\cD(\cA)}\|T^t(\rho)-\omega\|_1\leq\epsilon\}
\end{align}
and establish the following two-sided bounds on this quantum mixing time in terms of the generalized quantum Dobrushin coefficient:
\DobrushinTight*
\begin{proof}
    Let $d_s(\rho)=\frac{1}{2}\|T^s(\rho)-\omega\|_1$ and $\kappa_s=\kappa(T^s)$. For any $m$ and $s$, since $T(\omega)=\omega$, by \cref{eq: defn of quantum Dobrushin coefficient}
    \begin{align}
        d_{ms}(\rho)\leq\kappa_sd_{(m-1)s}(\rho)\leq\ldots\leq\kappa_s^m d_0(\rho)\leq e^{-m(1-\kappa_s)}.
    \end{align}
    Since $\kgen(T)<1$, there exists a finite $s$ such that $\kgen(T)=1-\frac{1-\kappa_s}{s}$. Choose $m=\lceil \frac{\log(\frac{2}{\epsilon})}{1-\kappa_s}\rceil$, then 
    \begin{align}
        \|T^{ms}(\rho)-\omega\|_1=2d_{ms}(\rho)\leq\epsilon.
    \end{align}
    Since $s\leq\frac{1}{1-\kgen(T)}$, we have
    \begin{align}
        t_{\mix}(T,\omega,\epsilon)\leq ms\leq s+\frac{s\log(\frac{2}{\epsilon})}{1-\kappa_s}\leq\frac{1+\log(\frac{2}{\epsilon})}{1-\kgen(T)}.
    \end{align}
    For the lower bound, let $t=t_{\mix}(T,\omega,\epsilon)$.
    Then by the triangle inequality,
    \begin{align}
        \kappa_t\leq 2\sup_{\rho\in\cD(\cA)}d_t(\rho)\leq\epsilon,
    \end{align}
    and therefore
    \begin{align}
        1-\kgen(T)\geq\frac{1-\kappa(T^t)}{t}\geq\frac{1-\epsilon}{t},
    \end{align}
    which concludes the proof.
\end{proof}

\subsection{Fourier analysis on finite abelian groups}
We follow the treatment in \cite{Rudin1990Fourier}. Let $(A,+)$ be a finite abelian group. Its character group is
\begin{align}
    \hat A\coloneqq \{\chi\colon A\to\bC^\times:\chi(a+b)=\chi(a)\chi(b)\}
\end{align}
with group operation given by $(\chi_1+\chi_2)(a)=\chi_1(a)\chi_2(a)$. For $f\colon A\to\bC$ define
\begin{align}
    \hat f(\chi)\coloneqq \frac1{|A|}\sum_{a\in A}f(a)\overline{\chi(a)}.
\end{align}
The Fourier inversion formula says
\begin{align}
    f(a)=\sum_{\chi\in\hat A}\hat f(\chi)\chi(a).
\end{align}
For probability measures $p,q$ on $\hat A$, denote
\begin{align}
    \hat p(a)\coloneqq \sum_{\chi\in\hat A}p(\chi)\chi(a),\qquad(p*q)(\eta)=\sum_{\chi+\psi=\eta}p(\chi)q(\psi).
\end{align}
Then
\begin{align}\label{eq:fourier-convolution-product}
  \widehat{p*q}(a)=\sum_\eta\sum_{\chi+\psi=\eta}p(\chi)q(\psi)\eta(a)=\sum_{\chi,\psi}p(\chi)q(\psi)\chi(a)\psi(a)=\hat p(a)\hat q(a).
\end{align}
We will make use of Bochner's Theorem for finite abelian groups:

\begin{thm}[Bochner]\cite[Theorem 19]{Bell2015Gelfand}\label{theorem: Bochner}
    A function $\alpha\colon A\to\bC$ is positive semidefinite (i.e., the matrix $K_\alpha(x,y)\coloneqq \alpha(x-y)$ is positive semidefinite) and satisfies $\alpha(0)=1$ if and only if there is a probability distribution $p$ on $\widehat A$ such that $\alpha=\hat{p}$.
    The distribution is unique and is given by $p=\hat{\alpha}$.
\end{thm}

\subsection{Graph states and hypergraph states} \label{subsection: hypergraph state intro}
A hypergraph is a pair $\Gamma=(V,E)$ where $V$ is a finite vertex set and $E$ is a set of nonempty subsets of $V$.  Elements $e\in E$ are called hyperedges. For a hyperedge $e\subseteq V$, define the gate $C_e$ by
\begin{align}
  C_e\ket{x}=(-1)^{\prod_{v\in e}x_v}\ket{x},
  \qquad x\in\bF_2^V
\end{align}
which are diagonal and commute with one another; then the hypergraph state $\ket{\Gamma}$ is defined by \cite{Rossi2013QuantumHypergraphStates}
\begin{align}
  \ket{\Gamma}\coloneqq \prod_{e\in E}C_e\,\ket{+}^{\otimes |V|}.
\end{align}

Graph states are the degree-$2$ part of the hypergraph-state family where all hyperedges have size $2$, which makes $C_e=CZ_e$. Graph states are stabilizer states because they are obtained from $|+\rangle^{\otimes n}$ by Clifford gates $CZ_e$. In contrast, a hyperedge gate of size $3$ or more, such as $\text{CCZ}$, is generally not a Clifford gate.  Therefore general hypergraph states are usually not stabilizer states, although they still have a generalized stabilizer formalism and are naturally described by phase polynomials \cite{Rossi2013QuantumHypergraphStates}.

\section{Sample complexity as a mixing time for geometrically uniform pure state ensembles} \label{section: GU ensemble}
Given a finite group $G$ with a unitary representation $U_g$ on a Hilbert space $\cH$ and a generator state $\ket{\psi}\in\cH$, we consider the geometrically uniform (GU) ensemble generated by $G$ and $|\psi\rangle$, that is, the uniform ensemble consisting of the states $\ket{\psi_g}\coloneqq U_g\ket{\psi}$ \cite{Eldar2004OptimalDetection, zhou2025distinguishability}. The Gram matrix of the ensemble $(\frac{1}{|G|},|\psi_g\rangle)_{g\in G}$ has matrix elements $X_{gh}=\bra{\psi}U(g^{-1}h)\ket{\psi}$ and the $t$-copy Gram matrix is given by $X_t=|G|\Phi_X^t(\tau)$, where $\Phi_X(Y)\coloneqq X\circ Y$ is the Hadamard channel defined by $X$ with $\circ$ denoting the Schur product and $\tau\coloneqq\frac{1}{|G|}\mathbb J$ is the normalized state of the all-one matrix $\mathbb J$. An important property of pure state GU ensembles is that the PGM \eqref{eq:pretty-good-measurement} is optimal and has the closed form formula given by $p_{\mathrm{PGM}}=\frac{1}{|G|^2}(\tr\sqrt{X})^2$ \cite{Eldar2004OptimalDetection, zhou2025distinguishability}, so the $t$-copy optimal success probability satisfies
\begin{align}
    \psucc^*(t)=p_{\mathrm{PGM}}(t)=\frac{1}{|G|^2}\left(\tr\sqrt{X_t}\right)^2=F(\frac{1}{|G|}X_t,\omega),
\end{align}
with $\omega\coloneqq\frac{1}{|G|}\one$ the completely mixed state. Therefore the sample complexity can be written as
\begin{align} \label{eq: sample complexity as fidelity mixing}
    k_{\min}(\epsilon)=\min\{t:F(\Phi_X^t(\tau),\omega)\geq1-\epsilon\},
\end{align}
which coincides with the fidelity-based mixing time of the Hadamard channel $\Phi_X$ for the fixed initial state $\tau$ and the stationary state $\omega$.

Let $H\coloneqq \{h\in G:U_h\ket{\psi}=\ket{\psi}\}$ be the stabilizer subgroup of the generator state $|\psi\rangle$. Since the states are indistinguishable within the same coset $gH$, the meaningful problem reduces to discriminating the states $\{\ket{\psi_{gH}}\}$ labeled by $G/H$. From now on we set $n=|G/H|$ and refer to $k_{\min}$ as the sample complexity of discriminating the states $\ket{\psi_{gH}}$. The Gram matrix $X$ now lives in $\End(\bC[G/H])$, and previous arguments for $k_{\min}(\epsilon)=\min\{t:F(\Phi_X^t(\tau_n),\omega_n)\geq1-\epsilon\}$ still apply. We can write $X=\sum_{j=0}^{d}\alpha_jA_j$ where $g_j$ are the double coset representatives and
\begin{align} \label{eq: matrix A_r}
    \alpha_j=\bra{\psi}U(g_j)\ket{\psi};\quad (A_j)_{xH,yH}\coloneqq 1\{x^{-1}y\in Hg_jH\}.
\end{align}
Note that $T\in\End_G(\bC[G/H])\Leftrightarrow T_{gxH,gyH}=T_{xH,yH}$ $\forall\:g,x,y\in G$,
and $\{(xH,yH):x^{-1}y\in Hg_jH\}$ are exactly the $G$-orbits of coset pairs, therefore we can write the Hecke algebra as
\begin{align}
    \cA\coloneqq\End_G(\bC[G/H])=\left\{T=\sum\nolimits_j c_jA_j:c_j\in\bC\right\}=\text{Span}\{A_j\}.
\end{align}
In particular, the Gram matrix $X\in\cA$ and $\Phi_X$ preserves $\cA$.

Now consider the quantum homogeneous mixing time \cite{George2026QuantumDoeblinCoefficients} of $\Phi_X$ for the fixed initial state $\tau_n$ and the stationary state $\omega_n$, which quantifies convergence in terms of trace distance and is defined by
\begin{align} \label{eq: mixing time for GU pure states}
    t_{\min}(\epsilon)=\min\{t:\|\Phi_X^t(\tau_n)-\omega_n\|_1\leq\epsilon\}.
\end{align}
By Fuchs-van-de-Graaf inequalities relating fidelity and trace distance, we have
\begin{align}
    \frac{d_t^2}{4}\leq 1-F(\Phi_X^t(\tau_n),\omega_n)\leq d_t-\frac{d_t^2}{4}
\end{align}
where $d_t\coloneqq\|\Phi_X^t(\tau_n)-\omega_n\|_1$.
Therefore,
\begin{align} \label{eq: sandwiched bound for GU sample complexity}
    t_{\min}(2\sqrt{\epsilon})\leq k_{\min}(\epsilon)\leq t_{\min}(2-2\sqrt{1-\epsilon}).
\end{align}
For GU ensembles, the optimal average success probability coincides with the optimal worst-case success probability: For any POVM $\{\Pi_{gH}\}$, the averaged measurement $\overline{\Pi}_{gH}\coloneqq\frac{1}{|G|}\sum_{k\in G}U_k\Pi_{k^{-1}gH}U_k^\dagger$ achieves the same average success probability with identical conditional success probability for every label. Therefore, by \Cref{proposition: worst-case sample complexity equivalence}, the two-sided bound \cref{eq: sandwiched bound for GU sample complexity} becomes a tight estimate $k_{\min}=\Theta_\epsilon(t_{\min})$, and the estimation of discrimination sample complexity reduces to a mixing time problem.

\subsection{Tight estimate by Quantum Dobrushin coefficient}
\label{sec:tight-estimate-by-quantum-dobrushin-coefficient}
We use the definitions of the quantum Dobrushin coefficient and mixing time in \Cref{subsubsection: Quantum Dobrushin coefficient}. To apply \cref{eq: worst case mixing time for generic channel} to study $t_{\min}$ and sample complexity, we consider the stationary state $\omega_n\coloneqq\frac{1}{n}\one\in\cA$ and the channel $T\coloneqq\Phi_X\big|_\cA$ where we use the fact that $\Phi_X$ is invariant on $\cA$ since $X\in\cA$. Then \Cref{theorem: quantum Dobrushin estimate} shows that, if $\epsilon\in(0,\frac{1}{2})$ and $\kgen(T)<1$, then
    \begin{align} \label{eq: GU Dobrushin sandwiched bound}
        \frac{1-\epsilon}{1-\kgen(T)}\leq t_{\mix}(T,\omega_n,\epsilon)\leq\frac{1+\log(\frac{2}{\epsilon})}{1-\kgen(T)}.
    \end{align}

This $t_{\mix}$ is a worst-case initial-state mixing time while the sample complexity is sandwiched bounded by the fixed-initial-state mixing time $t_{\min}$ in \cref{eq: sandwiched bound for GU sample complexity}. However, in the GU case the initial state $\tau_n=\frac{1}{n}\mathbb J$ happens to be the worst initial state:

\begin{prop} \label{prop: GU worst initial state}
    For every $t\in\mathbb{N}$,
    \begin{align}
        \sup_{\rho\in\cD(\cA)}\|T^t(\rho)-\omega_n\|_1=\|T^t(\tau_n)-\omega_n\|_1.
    \end{align}
\end{prop}
\begin{proof}
    For any $\rho\in\cD(\cA)$ we have $\rho_{xx}=\rho_{gx,gx}$ for all $g\in G,x\in G/H$, and therefore $\rho_{xx}=\frac{1}{n}$ for all $x\in G/H$. Let $Y=n\rho$ and $\Phi_Y$ be the Hadamard channel.
    Then $\Phi_Y(\omega_n)=\omega_n$ and $\Phi_X^t(\rho)=\Phi_Y\Phi_X^t(\tau_n)$, and therefore
    \begin{align}
        \|\Phi_X^t(\rho)-\omega_n\|_1=\|\Phi_Y\Phi_X^t(\tau_n)-\Phi_Y(\omega_n)\|_1\leq\|\Phi_X^t(\tau_n)-\omega_n\|_1,
    \end{align}
    which proves the claim.
\end{proof}

Finally, we classify the families for which $\kgen(T)<1$:
\begin{prop} \label{prop: gamma(X)<1}
    $\kgen(T)<1\Longleftrightarrow\max_{g\notin H}|\bra{\psi}U_g\ket{\psi}|<1$.
\end{prop} 
\begin{proof}
    Let $\mu(X)\coloneqq\max_{x\neq y\in G/H}|X_{xy}|=\max_{g\notin H}|\bra{\psi}U_g\ket{\psi}|$. Any $B\in\cA$ has constant diagonal, so it has zero diagonal if $\tr(B)=0$.
    Hence,
    \begin{align}
        \|T^s(B)\|_1\leq\sqrt{n}\|T^s(B)\|_2=\sqrt{n}\sqrt{\sum_{x\neq y\in G/H}|X_{xy}|^{2s}|B_{xy}|^2}\leq\sqrt{n}\mu(X)^s\|B\|_2\leq\sqrt{n}\mu(X)^s\|B\|_1.
    \end{align}
    Therefore, when $\mu(X)<1$, we can choose $s$ sufficiently large such that $\sqrt{n}\mu(X)^s<1$, and
    \begin{align}
        \kappa(T^s)=\sup_{\substack{B=B^\dagger\in\cA\\\tr(B)=0\\B\neq 0}}\frac{\|T^s(B)\|_1}{\|B\|_1}\leq\sqrt{n}\mu(X)^s<1\Rightarrow\kgen(T)<1.
    \end{align}
    Conversely, suppose $|X_{xy}|=1$ for some $x\neq y\in G/H$. Let $B_0=\mathbb J-\one$ with zero diagonal where $\mathbb J$ is the all-one matrix.
    Then $T^{ms}(B_0)=X^{\circ(ms)}-\one$ has absolute value 1 at position $(x,y)$ for all $m$ and $s$, so $\|T^{ms}(B_0)\|_1\geq1$.
    If $\kappa(T^s)<1$ for some $s$, then by the contraction property of trace distance we get $\|T^{ms}(B_0)\|\leq\kappa(T^s)^m\|B_0\|_1\to0$ for $m\to\infty$, and thus a contradiction, showing that $\kgen(T)=1$.
\end{proof}

Combining \cref{eq: sandwiched bound for GU sample complexity}, \cref{eq: GU Dobrushin sandwiched bound}, \Cref{prop: GU worst initial state} and \Cref{prop: gamma(X)<1}, we have proved: 
\GUDobrushin*

\subsection{Mixing of classical random walk on dual hypergroup} \label{subsection: classical random walk}
Suppose $\bC[G/H]\cong\bigoplus_{\pi\in\hat{G}}V_\pi\otimes M_\pi$ is a complete decomposition of irreducible representations.
Then the algebra of $G$-invariant operators decomposes as
\begin{align}
    \cA=\End_G(\bC[G/H])\cong\bigoplus_{\pi\in\hat{G}}\bC\otimes\End(M_\pi).
\end{align}
Therefore the Hecke algebra $\End_G(\bC[G/H])$ is commutative if and only if $\dim M_\pi\leq 1$ for all $\pi\in\hat{G}$, and in this case $(G,H)$ is called a Gelfand pair \cite{CeccheriniSilbersteinScarabottiTolli2008, ScarabottiTolli2013}. For each $\pi\in\hat{G}$ with $\dim M_\pi=1$, let $E_\pi$ be the projector onto $V_\pi$, then $\{E_\pi:\dim M_\pi=1\}$ form an orthogonal basis of the commutative Hecke algebra $\cA\cong\bigoplus_\pi\bC E_\pi$. These projectors are also called primitive idempotents of the algebra. We denote by $d+1=\dim\cA$ the number of primitive idempotents.

The following treatment of dual hypergroups follows \cite{Godsil2010AssociationSchemes,Voit2019ContinuousAssociationSchemes,Voit2017GeneralizedCommutativeAssociationSchemes}. A finite hypergroup $(\Omega,*)$ is a finite set $\Omega$ together with an associative bilinear convolution law on probability measures such that:\\
(a) for each $x,y\in\Omega$, the convolution $\delta_x*\delta_y$ is a probability measure on $\Omega$;\\
(b) there is an identity element $e\in\Omega$ with $\delta_e*\delta_x=\delta_x*\delta_e=\delta_x$;\\
(c) there is an involution map compatible with the convolution.

In particular, when $\cA$ is commutative, it admits a dual hypergroup $\hat{\Omega}=\{0,...,d\}$ with convolution
\begin{align}
    \delta_i * \delta_j \coloneqq  \sum_{l=0}^d K_{ij}^{l}\,\delta_l;\quad\quad K_{ij}^l\coloneqq \frac{m_l}{m_im_j}q_{ij}^l
\end{align}
where $m_i\coloneqq\text{rank}(E_i)=\tr(E_i)$ and $q_{ij}^{k}$ are the Krein parameters defined via
\begin{align}
    E_i\circ E_j=\frac{1}{n}\sum_{k=0}^{d}q_{ij}^{k}E_k,
\end{align}
with $\circ$ being the Schur product.

Define a map $\Xi\colon \mc P(\hat{\Omega})\to\cD(\cA)$, where $\mc P(\hat\Omega)$ denotes the space of probability measures on $\hat\Omega$, by
\begin{align}
    \Xi(q)\coloneqq\sum_i\frac{q_i}{m_i}E_i.
\end{align}
This map is a bijection. It preserves the trace distance:
\begin{align} \label{eq: hypergroup bijection TV distance}
        d_{\TV}(q,r)=\frac{1}{2}\|q-r\|_1=\frac{1}{2}\sum_im_i \left|\frac{q_i-r_i}{m_i}\right|=\frac{1}{2}\|\Xi(q)-\Xi(r)\|_1.
\end{align}
Moreover, for any positive $Y=\sum_j\theta_jE_j\in\cA$ with $\tr(Y)=n$, the coefficient of $E_k$ in $\Phi_Y(\Xi(q))$ is 
\begin{align}
    \frac{1}{n}\sum_{i,j}\frac{q_i}{m_i}\theta_jq_{ij}^{k}=\frac{(q*p(Y))(k)}{m_k}
\end{align}
where $p(Y)$ is the distribution given by $p(Y)\coloneqq\Xi^{-1}(\frac{Y}{n})$, i.e. $p(Y)_j\coloneqq\frac{m_j\theta_j}{n}$. It follows that
\begin{align} \label{eq: hypergroup bijection convolution}
    \Phi_Y(\Xi(q))=\Xi(q*p(Y)).
\end{align}

Given any probability measure $p$ on $\hat{\Omega}$, the convolution induces a Markov kernel defined by
\begin{align} \label{eq: Markov chain K_p}
    K_p(i,j)\coloneqq (\delta_i*p)(j).
\end{align} 
In particular, $p^{*t}$ is the distribution of the random walk $K_p$ starting at $0$ after $t$ steps, with stationary distribution $\pi\coloneqq\{\pi_j=\frac{m_j}{n}\}$. In particular, by \cref{eq: hypergroup bijection TV distance} and \cref{eq: hypergroup bijection convolution},
\begin{align} \label{eq: trace distance to total variation}
    \|\Phi_X^t(\tau_n)-\omega_n\|_1=\|\Xi(p(X)^{*t})-\Xi(\pi)\|_1=2d_{TV}(p(X)^{*t},\pi),
\end{align}
therefore $t_{\min}$ is equal to the mixing time of the classical random walk given by $K_{p(X)}$ where $X$ is the Gram matrix of the GU ensemble. Also note that the spectrum of $K_{p(X)}$ consists exactly of $\alpha_j=\bra{\psi}U(g_j)\ket{\psi}$. We summarize these results in the following theorem:

\ClassicalMixing*

\begin{rem}
    The above analysis shows the equivalence between the sample complexity $k_{\min}$ and a classical total-variation mixing time only up to $\epsilon$. From the perspective of fidelity-type mixing time in \cref{eq: sample complexity as fidelity mixing}, 
    \begin{align}
        \psucc^*(t)=\frac{1}{n^2}\|X_t\|_\frac{1}{2}=\bigg(\sum_j\sqrt{p^{*t}(X)_j\pi_j}\bigg)^2
    \end{align}
    is the squared classical fidelity between the $t$-step walk and its stationary distribution. Therefore, the sample complexity of state discrimination coincides with a classical fidelity mixing time on the dual hypergroup $\hat\Omega$.
\end{rem}

\begin{ex}
    Quantum coupon collector problem \cite{Arunachalam2020QuantumCouponCollector}: Given $r<l$, consider the set $\Omega_{r,l}$ of $r$-subsets of $[l]\coloneqq\lbrace 1,\dots,l\rbrace$. The goal is to discriminate the states
    \begin{align}
        \ket{\psi_S}\coloneqq \frac{1}{\sqrt{r}}\sum_{i\in S}\ket{i} \quad \text{for $S\in\Omega_{r,l}$}.
    \end{align}
    We can write it as a GU ensemble generated by $G=S_l$ and $\ket{\psi_{S_0}}$ where $S_0=\{1,...,r\}$.
    Letting $H=S_r\times S_{l-r}$ be the stabilizer subgroup of $S_0$, the pair $(G,H)$ is a Gelfand pair. The Gram matrix $X$ has components $X_{ST}=\frac{|S\cap T|}{r}=1-\frac{j}{r}$ where $j\coloneqq r-|S\cap T|$, so $X=\sum_{j=0}^{d}(1-\frac{j}{r})A_j$ where $A_j$ is defined in \cref{eq: matrix A_r} and $d\coloneqq \min(r,l-r)$. Simple calculation shows $X$ has eigenvalues $\theta_0=\binom{l-1}{r-1}$ with multiplicity 1 and $\theta_1=\frac{1}{r}\binom{l-2}{r-1}$ with multiplicity $l-1$. By the standard calculations for Johnson schemes in \cite[Section 3.3]{Arunachalam2020QuantumCouponCollector}, $m_j=\tr(E_j)=\binom{l}{j}-\binom{l}{j-1}$, in particular $m_0=1$ and $m_1=l-1$, so $X=\theta_0 E_0+\theta_1 E_1$. Therefore $p=p(X)$ defined by \cref{eq: the map p} is supported on $\{0,1\}$, and $K_p$ defined in \cref{eq: Markov chain K_p} is a birth-death chain with spectrum $\alpha_j=1-\frac{j}{r}$ for $0\leq j\leq d$. By \Cref{prop: birth-death mixing},
    \begin{align} \label{eq: coupon collector estimate}
        t_{\mix}(K_p,\epsilon)=\Theta_\epsilon(\sum_{j=1}^{d}\frac{1}{1-\alpha_j})=\Theta_\epsilon(r\log d),
    \end{align}
    recovering the result $k_{\min}=\Theta_\epsilon(r\log\min\{r,l-r\})$ in \cite{Arunachalam2020QuantumCouponCollector}.
\end{ex}

\subsection{Tight estimate by classical Dobrushin coefficient}
By \cref{eq: hypergroup bijection TV distance} and \cref{eq: hypergroup bijection convolution},
\begin{align}
    \kappa(T^s)
    =\frac{1}{2}\sup_{\rho,\sigma\in\cD(\cA)}\|\Phi_X^s(\rho)-\Phi_X^s(\sigma)\|_1&=\frac{1}{2}\sup_{q,r\in\mc P(\hat\Omega)}\|\Phi_X^s(\Xi(q))-\Phi_X^s(\Xi(r))\|_1\\
    &=\sup_{q,r\in\mc P(\hat\Omega)}d_{\TV}(q*p(X)^{*s},r*p(X)^{*s}).
\end{align}
Denote $p=p(X)$ and let 
\begin{align}
    P_s(i,a)=\big(\delta_i* p^{*s}\big)(a).
\end{align}
Then, for any distributions $q,r$ on $\hat{\Omega}$, 
\begin{align}
    qP_s(j)=\sum_i q_iP_s(i,j)=\big((\sum_iq_i\delta_i)*p^{*s}\big)(j)=(q*p^{*s})(j)    
\end{align}
and
\begin{align}
    d_{\TV}(qP_s,rP_s)\leq\sum_{i,j\in\hat\Omega}q_ir_jd_{\TV}(P_s(i,\cdot),P_s(j,\cdot))\leq\max_{i,j\in\hat\Omega}d_{\TV}(P_s(i,\cdot),P_s(j,\cdot)).
\end{align}
The equality in this chain holds when we take $q=\delta_i,r=\delta_j$ for a maximizing pair $(i,j)$ on the RHS, so
\begin{align} \label{eq: quantum Dobrushin coeff is classical}
    \kappa(T^s)=\sup_{(i,j)\in\hat{\Omega}^2}d_{\TV}(P_s(i,\cdot),P_s(j,\cdot))=\kappa(P_s).
\end{align}
Therefore the quantum Dobrushin coefficient of the restricted Schur channel reduces exactly to the classical Dobrushin coefficient of the dual-hypergroup random walk, and so does the generalized quantum Dobrushin coefficient. Written explicitly in terms of representation theoretic quantities,
\begin{align} \label{eq: Gelfand Dobrushin formula}
    \kgen(T)=1-\sup_{s\in\mathbb N}\frac{1}{s}\min_{i,j\in\hat\Omega^2}\sum_{k\in\hat\Omega}\min \left\{ \sum_{l\in\hat\Omega}\frac{m_kq_{il}^{k}}{m_im_l}p^{*s}(l),\sum_{l\in\hat\Omega}\frac{m_kq_{jl}^{k}}{m_jm_l}p^{*s}(l) \right\}.
\end{align}
In this way, \Cref{theorem: GU Dobrushin estimate} gives a tight estimate of the state discrimination sample complexity in terms of the Krein parameters and dimensions of the primitive idempotents.

\section{Phase state case} \label{section: phase state}
In this section, we consider a simple family of Gelfand pairs, then specialize to a family of phase states that covers many examples in the literature. Let $A$ be a finite abelian group, $G=A\rtimes K$ where $K\leq\text{Aut}(A)$, and $H=K$. Since $H(a,k)H=H(a,1)H$, the $H$-double cosets are labeled by $K$-orbits in $A$, so the Hecke algebra is isomorphic to $\bC[A]^K$, which is commutative since $A$ is abelian. For any $K$-orbit $D\subset\hat{A}$, the operator $E_D$ on $\bC[A]$ defined by
\begin{align} \label{eq: phase state primitive idempotent}
    E_D(x,y)\coloneqq \frac{1}{|A|}\sum_{\chi\in D}\chi(y-x)
\end{align}
is a primitive idempotent of rank $|D|$. The dual hypergroup is isomorphic to $\hat{A}/K$ with convolution
\begin{align}
    \delta_{[\chi]}*\delta_{[\psi]}=\frac{1}{|K|}\sum_{k\in K}\delta_{[\chi(k\cdot\psi)]}.
\end{align}

Now consider the GU ensemble generated by group $G$ and generator $\ket{\psi_0}$ with stabilizer subgroup $H$. By \Cref{theorem: Bochner} there exists a unique probability measure $p$ on $\hat{A}$ such that
\begin{align}
    \alpha(a)\coloneqq \bra{\psi_0}U_a\ket{\psi_0}=\hat{p}(a)\quad\forall\:a\in A,
\end{align}
and $p=\hat{\alpha}$ is $K$-invariant.

\PhaseState*
\begin{proof}
    Let $\tilde{p}$ be the pushforward of $p$ to the dual hypergroup $\hat{A}/K$. Since $p(\chi)=\frac{\tilde{p}(D)}{|D|}$ for all $\chi\in D$ for any $K$-orbit $D\subset\hat{A}$, we have
    \begin{align}
        \alpha(a)=\sum_{\chi\in\hat{A}}p(\chi)\chi(a)=\sum_j\frac{\tilde{p}(D_j)}{|D_j|}\sum_{\chi\in D_j}\chi(a).
    \end{align}
    Since the Gram matrix $M(x,y)=\bra{\psi_0}U_{y-x}\ket{\psi_0}=\alpha(y-x)$ lives in the Hecke algebra, it can be written as a linear combination of primitive idempotents. Writing $a=y-x$ and using \cref{eq: phase state primitive idempotent} we have
    \begin{align}
        \alpha(a)=M(x,y)=\sum_j\theta_j E_j(x,y)=\sum_j\frac{\theta_j}{|A|}\sum_{\chi\in D_j}\chi(a).
    \end{align}
    Comparing the two equations above, one can see that
    \begin{align}
        \tilde{p}(D_j)=\frac{|D_j|\theta_j}{|A|}=\frac{m_j\theta_j}{n},
    \end{align}
    therefore $\tilde{p}$ is exactly the hypergroup distribution induced by the Gram matrix under the map \eqref{eq: the map p}. Denote $\pi=\Unif(\hat{A})$, by \cref{eq: trace distance to total variation},
    \begin{align}
        t_{\min}(\epsilon)=\min\{t:\|\tilde{p}^{*t}-\tilde{\pi}\|_1\leq\epsilon\}=\min\{t:\|p^{*t}-\pi\|_1\leq\epsilon\}.
    \end{align}
    Finally, by induction
    \begin{align}
        P(Z_t=\chi)=\sum_{\eta\in\hat{A}}P(Z_{t-1}=\eta)P(X_t=\chi-\eta)=\sum_{\eta\in\hat{A}}p^{*(t-1)}(\eta)p(\chi-\eta)=p^{*t}(\chi).
    \end{align}
    Therefore,
    \begin{align}
        t_{\min}(\epsilon)=\min\{t:\|p^{*t}-\pi\|_1\leq\epsilon\}=\min\{t:\|\law(Z_t)-\pi\|_1\leq\epsilon\}.
    \end{align}
\end{proof}

Now we specialize to a family of phase states where the random walk $Z_t$ is more explicit: Let $\mc X$ be a finite set with probability measure $\nu$, and $\Phi\colon \mc X\to\hat A$ be a feature map. Consider the GU ensemble generated by
    \begin{align} \label{eq: GU phase state with feature map}
        \ket{\psi_0}=\sum_{x\in\mc X}\sqrt{\nu(x)}\ket{x},\qquad U_a\ket{x}=\Phi(x)(a)\ket{x} \quad \text{for $a\in A,x\in\mc X$}.
    \end{align}
    Then
    \begin{align}
        \alpha(a)=\langle\psi_0|U_a|\psi_0\rangle=\sum_{x\in\mc X}\nu(x)\Phi(x)(a),
    \end{align}
    so $p=\hat{\alpha}=\Phi_\#\nu$ is the pushforward measure of $\nu$ under $\Phi$, and hence the lifted dual walk is 
    \begin{align} \label{eq:feature-lifted-walk-expanded}
    Z_t=\Phi(X_1)+\cdots+\Phi(X_t),\qquad X_1,X_2,\ldots\stackrel{\rm i.i.d.}{\sim}\nu.
    \end{align}
    The generalized Dobrushin coefficient of this walk can be explicitly written in terms of $A,\nu,\Phi$:
    \begin{align}
        \kgen(A,\nu,\Phi)=1-\sup_{s\in\mathbb N}\frac{1}{s}\min_{a\in\hat A}\sum_{z\in\hat A}\min\{w_s(z),w_s(z-a)\}
    \end{align}
    where
    \begin{align}
        w_s(z)\coloneqq\sum_{\substack{x_1,\ldots,x_s\in\mc X\\\sum_i\Phi(x_i)=z}}\prod_{j=1}^{s}\nu(x_j).
    \end{align}
    In particular, if $\nu$ is uniform, then
    \begin{align} \label{eq: phase state Dobrushin}
        \kgen(A,\Phi)=1-\sup_{s\in\mathbb N}\frac{1}{s|\mc X|^s}\min_{a\in\hat A}\sum_{z\in\hat A}\min\{R_s(z),R_s(z-a)\}
    \end{align}
    where
    \begin{align}
        R_s(z)\coloneqq\# \left\{(x_1,\ldots,x_s)\in\mc X^s:\sum\nolimits_{i=1}^{s}\Phi(x_i)=z\right\}.
    \end{align}
    Note that this is the generalized Dobrushin coefficient of the lifted walk, which in general differs from that in \cref{eq: Gelfand Dobrushin formula}. However, for the purpose of estimating the sample complexity, both yield the same order up to $\epsilon$.

\subsection{Triangular feature map}
A powerful framework to exploit the structure of the feature map comes from Knothe–Rosenblatt triangular transport \cite{ramgraber2025friendly}. Such maps are standard in probability theory because they respect the natural filtration generated by the coordinates. In particular, our application of the framework converts a global high-dimensional mixing problem into a sequence of low-dimensional conditional Fourier estimates.

Suppose the dual group $\hat{A}$ admits a decomposition $\hat{A}=B_1\times\ldots\times B_n$ and the feature map takes the triangular form
\begin{align} \label{eq: triangular feature map}
    \Phi(x)=(\Phi_1(x^1),\Phi_2(x^1,x^2),\ldots,\Phi_n(x^1,\ldots,x^n))     
\end{align}
where $\Phi_i(x^{\leq i})\in B_i$. Then $Z_t=(Z_{t,1},\ldots,Z_{t,n})$ where
\begin{align}
    Z_{t,i}=\sum_{j=1}^{t}\Phi_i(X_j^{\leq i}).
\end{align}
The following theorem provides an upper bound for the total variation distance and thus $t_{\min}$, which will be used later in our discussion of the degree-$d$ phase states.

\Triangular*
\begin{proof}
    Let $\cF_i=\sigma(X_j^1,\ldots,X_j^i:1\leq j\leq t)$ denote the sigma-algebra generated by $X_j^i$, then $Z_{t,1},\ldots,Z_{t,i}$ are $\cF_i$-measurable. Defining the probability measure $\mu_i(z)\coloneqq P(Z_{t,i}=z|\cF_{i-1})$, by \Cref{lemma: total variation chain inequality},
    \begin{align}
        d_{\TV}(\law(Z_t),\Unif(\hat{A}))\leq\sum_{i=1}^{n}\bE d_{\TV}(\mu_i,\Unif(B_i)).
    \end{align}
    Applying Fourier transformation we get
    \begin{align}
        \hat\mu_i(\chi)=\sum_{z\in B_i}\mu_i(z)\chi(z)=\bE\left[\sum_{z\in B_i}\chi(z)1_{Z_{t,i}=z}|\cF_{i-1}\right]=\bE[\chi(Z_{t,i})|\cF_{i-1}]=\prod_{j=1}^{t}\bE\left[\chi(\Phi_i(X_j^{\leq i}))|\cF_{i-1} \right],
    \end{align}
    so
    \begin{align}
        \bE[|\hat{\mu}_i(\chi)|^2]=\prod_{j=1}^{t}\bE\left[|\bE[\chi(\Phi_i(X_j^{\leq i})|\cF_{i-1})]|^2\right]\leq\lambda_i^t.
    \end{align}
    By the Cauchy-Schwarz inequality and Parseval's identity, 
    \begin{align}
        \bE d_{\TV}(\mu_i,\Unif(B_i))\leq\frac{1}{2}\bE\left(\sum\nolimits_{\chi\neq 1}|\hat\mu_i(\chi)|^2\right)^\frac{1}{2}\leq\frac{1}{2}\left(\sum\nolimits_{\chi\neq 1}\bE|\hat\mu_i(\chi)|^2\right)^\frac{1}{2}\leq\frac{1}{2}\sqrt{(|B_i|-1)\lambda_i^t}.
    \end{align}
    Therefore,
    \begin{align}
        d_{\TV}(\law(Z_t),\Unif(\hat{A}))\leq\frac{1}{2}\sum_{i=1}^{n}\sqrt{(|B_i|-1)\lambda_i^t}.
    \end{align}
\end{proof}

\subsection{Degree-$d$ phase states over $\bF_q^n$} \label{subsection: degree-d phase states}
We follow the definition in \cite{alrabiah2026nolowdegree}. Let $q$ be a prime, $\omega_q=e^{\frac{2\pi i}{q}}$, and $n,d\geq 1$ be integers. For any polynomial $f\colon \bF_q^n\to\bF_q$ of degree at most $d$, the degree-$d$ phase state $|\psi_f\rangle$ is defined as
\begin{align}
    \ket{\psi_f}=q^{-\frac{n}{2}}\sum_{x\in\bF_q^n}\omega_q^{f(x)}\ket{x}.
\end{align}
We study the sample complexity of discriminating the uniform state ensemble consisting of $\ket{\psi_f}$ in this section. In particular, define
\begin{align}
    I_{n,d,q}\coloneqq\{\alpha\in\{0,\ldots,q-1\}^n:1\leq|\alpha|\leq d\},\qquad m\coloneqq|I_{n,d,q}|,
\end{align}
and
\begin{align}
    P_{n,d,q}= \left\{f_a(x)=\sum\nolimits_{\alpha\in I_{n,d,q}}a_\alpha x^\alpha:a\in F_q^m \right\}.
\end{align}
Since $\ket{\psi_{f+c}}=\omega_q^c\ket{\psi_f}$, it suffices to assume $f\in P_{n,d,q}$. The work of  \textcite{arunachalam_et_al:LIPIcs.TQC.2023.3} showed that $\Theta(n^{d-1})$ copies are necessary and sufficient for learning a degree-$d$ phase state over $\bF_2^n$, and \textcite{alrabiah2026nolowdegree} conjectured the same order of upper bound for the sample complexity over $\bF_q^n$. We will prove this conjecture in this subsection by applying \Cref{thm: total variation upper bound for triangular feature map} to upper-bound $t_{\min}(\epsilon)$.

The degree-$d$ phase states form a GU ensemble in the form of \cref{eq: GU phase state with feature map} with feature map given by $\Phi(x)=(x^\alpha)_{\alpha\in I_{n,d,q}}$ and $\nu=\Unif(\bF_q^n)$. We can write the feature map in the triangular form of \cref{eq: triangular feature map}: Set 
\begin{align}
    I_i\coloneqq\{\alpha\in I_{n,d,q}:\alpha_i\geq1,\alpha_j=0\:\forall\:j>i\},
\end{align}
then $I_{n,d,q}=I_1\sqcup\ldots\sqcup I_n$ and $\hat{A}=B_1\times\ldots\times B_n$ where $B_i\cong\bF_q^{|I_i|}$. Every monomial $x^\alpha$ for $\alpha\in I_i$ can be written as $x^\alpha=(x^1)^{\beta_1}\ldots(x^{i-1})^{\beta_{i-1}}(x^i)^l$ for some $(\beta_1,\ldots,\beta_{i-1},l)\in J_i$, where 
\begin{align}
    J_i\coloneqq\{(\beta,l):0\leq\beta_j\leq q-1,1\leq l\leq q-1,|\beta|+l\leq d\}.
\end{align}
Hence, $\Phi_i$ in the form of \cref{eq: triangular feature map} can be written as
\begin{align}
    \Phi_i(x)=\big((x^1)^{\beta_1}\ldots(x^{i-1})^{\beta_{i-1}}(x^i)^l\big)_{(\beta,l)\in J_i}.
\end{align}
For any $b=(b_{\beta,l})_{(\beta,l)\in J_i}\in B_i$, the corresponding character is $\chi_b(z)=\exp\left(\log(\omega_q)\sum_{(\beta,l)\in J_i}b_{\beta,l}z_{\beta,l}\right)$, so
\begin{align}
    \lambda_i=\max_{b\in B_i\backslash\{0\}}\bE_{X^{<i}}|\bE_{X^i}\chi_b(\Phi_i(X^{\leq i}))|^2&=\max_{b\in B_i\backslash\{0\}}\bE_{X^{<i}}\bigg|\frac{1}{q}\sum_{c\in\bF_q}\chi_b(\Phi_i(X^{<i},c))\bigg|^2\\
    &=\max_{b\in B_i\backslash\{0\}}\bE_{X^{<i}}\bigg|\frac{1}{q}\sum_{c\in\bF_q}\omega_q^{\sum_{(\beta,l)\in J_i}b_{\beta,l}(X^{<i})^\beta c^l}\bigg|^2.
    \label{eq:lambda-i}
\end{align}
To find a universal upper bound for $\lambda_i$, we prove the following two auxiliary lemmas:

\begin{lem} \label{lemma: polynomial anti-concentration}
    For any nonzero $h\in P_{n,d,q}$, we have 
    \begin{align}
        P_{Y\sim\Unif(\bF_q^n)}(h(Y)\neq 0)\geq q^{-d}. 
    \end{align}
    Moreover, if $R$ is any commutative ring and $h$ is a nonzero multilinear polynomial over $R$ of degree at most $d$, then
    \begin{align}
        P_{Y\sim\Unif(\{0,1\}^n)}[h(Y)\neq0]\geq 2^{-d}.
    \end{align}
\end{lem}
\begin{proof}
    We use induction on the number of variables. Writing $h(y,z)=\sum_{j=0}^{a}h_j(y)z^j$ where $h_a\neq0$ and $a\leq q-1$, it suffices to show the claim for the case $a>0$. Then $h_a$ has degree at most $d-a$, and whenever $h_a(y)\neq0$, the polynomial in $z$ is nonzero for at least one of the $q$ values of $z$, so
    \begin{align}
        P(h(Y,Z)\neq 0)\geq q^{-1}P(h_a(Y)\neq 0)\geq q^{-1-(d-a)}\geq q^{-d},
    \end{align}
    proving the first claim. For the second claim, write $h(y,z)=h_0(y)+zh_1(y)$ with $h_1\neq0$. Whenever $h_1(y)\neq0$, at least one of $h(y,0)$ and $h(y,1)$ is nonzero, therefore
    \begin{align}
        P(h(Y,Z)\neq0)\geq \frac12\Pr[h_1(Y)\neq0],
    \end{align}
    and induction on the degree gives the claim.
\end{proof}

\begin{lem} \label{lemma: upper bounding phase sum}
    If $g\colon \bF_q\to\bF_q$ satisfies $g(0)=0$ and is nonzero, then 
    \begin{align}
        \bigg|\sum_{c\in\bF_q}\omega_q^{g(c)}\bigg|^2\leq q^2-4\sin^2(\frac{\pi}{q}).
    \end{align}
\end{lem}
\begin{proof}
    Let $c_0$ be such that $g(c_0)\neq0$, then $|\omega_q^{g(0)}-\omega_q^{g(c_0)}|^2\geq 4\sin^2(\frac{\pi}{q})$, therefore 
    \begin{align}
        \bigg|\sum_{c\in\bF_q}\omega_q^{g(c)}\bigg|^2=q^2-\sum_{c<c'}|\omega_q^{g(c)}-\omega_q^{g(c')}|^2\leq q^2-4\sin^2\left(\frac{\pi}{q}\right),
    \end{align}
    and we are done.
\end{proof}

Now for every $b\neq 0$, we can choose $l$ such that $h_l^{(b)}(y)=\sum_{\beta:|\beta|+l\leq d}b_{\beta,l}y^\beta$ is nonzero with degree at most $d-l$.
Then by \Cref{lemma: polynomial anti-concentration}, 
\begin{align}
    P(h_l^{(b)}(Y)\neq0)\geq q^{-(d-l)}.
\end{align}
Let $g_{b,y}(c)=\sum_{l=1}^{\min(q-1,d)}h_l^{(b)}(y)c^l$, then $g_{b,y}(0)=0$ and it is nonzero whenever $h_l^{(b)}(Y)\neq0$.
Thus, \Cref{lemma: upper bounding phase sum} yields
\begin{align}
    \bE_{X^{<i}}\bigg|\frac{1}{q}\sum_{c\in\bF_q}\omega_q^{g_{b,X^{<i}}(c)}\bigg|^2&\leq P\big(h_l^{(b)}(X^{<i})=0\big)+P\big(h_l^{(b)}(X^{<i})\neq0\big)\left(1-4\frac{\sin^2(\frac{\pi}{q})}{q^2}\right)\\
    &\leq 1-4\frac{\sin^2(\frac{\pi}{q})}{q^{d+1}}\eqqcolon\lambda_{q,d}
\end{align}
for any $b\neq 0$, and using \eqref{eq:lambda-i} we get $\lambda_i\leq\lambda_{q,d}$ for all $i$. By \Cref{thm: total variation upper bound for triangular feature map}, setting $r=\max_i|I_i|\log q$,
\begin{align} \label{eq: degree-d phase state estimate}
    t_{\min}(\epsilon)\leq\frac{r+2\log(\frac{n}{\epsilon})}{-\log\lambda_{q,d}}\leq O_{q,d,\epsilon}(n^{d-1})
\end{align}
since $r\leq\frac{dm}{n}\log q$ and $m\leq\binom{n+d}{d}$.

\subsection{Generalized degree-$d$ Boolean phase states over $\mathbb Z_q$} \label{subsection: Boolean phase states}
We follow the definition in the work of \textcite{arunachalam_et_al:LIPIcs.TQC.2023.3}, which left the tight upper bound for the sample complexity of learning generalized phase states with entangled measurements as an open problem. Fix $d\geq2$ and an even integer $q\geq2$, and let $m=\sum_{j=1}^{\min(n,d)}\binom{n}{j}$ and $A=(\mathbb Z_q)^m$. For $a=(a_S)_{S\in\bold S}$ where $\bold S\coloneqq\{S\subset[n]:1\leq|S|\leq d\}$, define the generalized degree-$d$ Boolean phase states over $\mathbb Z_q$ by
\begin{align}
    \ket{\psi_a}\coloneqq2^{-\frac{n}{2}}\sum_{x\in\{0,1\}^n}\omega_q^{f_a(x)}\ket{x},\qquad f_a(x)\coloneqq\sum_{S\in\bold S}a_Sx_S(\text{mod }q).
\end{align}
These states form a GU ensemble in the form of \cref{eq: GU phase state with feature map} with feature map given by $\Phi(x)=(x_S)_{1\leq|S|\leq d}$ and $\nu=\Unif$. We can write the feature map in the triangular form of \cref{eq: triangular feature map}: With
\begin{align}
    \bold S_i &\coloneqq\{S\subset[i-1]:|S|\leq d-1\}\\
    r_i &\coloneqq|\bold S_i|=\sum_{j=0}^{d-1}\binom{i-1}{j},
\end{align}
we have $\bold S=\bigsqcup_i\{S\cup\{i\}:S\in\bold S_i\}$ and $\hat{A}=B_1\times\ldots\times B_n$ where $B_i\cong(\mathbb Z_q)^{|\bold S_i|}$. The feature map can be written by
\begin{align}
    \Phi_i(x)=(x_ix_S)_{S\in\bold S_i}.
\end{align}
For any $b=(b_S)_{S\in\bold S_i}\in B_i$, the corresponding character is $\chi_b(z)=\exp\left(\log(\omega_q)\sum_{S\in\bold S_i}b_Sz_S\right)$, so
\begin{align}
    \lambda_i=\max_{b\in B_i\backslash\{0\}}\bE_{X^{<i}}\bigg|\frac{1+\omega_q^{h_b(X^{<i})}}{2}\bigg|^2=\max_{b\in B_i\backslash\{0\}}\bE_{X^{<i}}\cos^2\left(\frac{\pi h_b(X^{<i})}{q}\right)
\end{align}
where $h_b(y)\coloneqq\sum_{S\in\bold S_i}b_Sy_S(\text{mod }q)$. Since $b_S=\sum_{T\subset S}(-1)^{|S|-|T|}h(1_T)\quad(\text{mod }q)$ and $b\neq0$, we have $h_b\neq0$. By \Cref{lemma: polynomial anti-concentration}, $P(h_b(Y)\neq0)\geq2^{1-d}$, so
\begin{align}
    \lambda_i=\max_{b\in B_i\backslash\{0\}}\bE_{Y}\cos^2\left(\frac{\pi h_b(Y)}{q}\right)\leq1-\min_b P(h_b(Y)\neq0)\sin^2\left(\frac{\pi}{q}\right)\leq1-2^{1-d}\sin^2\left(\frac{\pi}{q}\right)\eqqcolon\lambda_{q,d}.
\end{align}
By \Cref{thm: total variation upper bound for triangular feature map}, setting $r=\max_i\log|B_i|$,
\begin{align} \label{eq: Boolean phase state estimate}
    t_{\min}\leq\frac{r+2\log(\frac{n}{\epsilon})}{-\log\lambda_{q,d}}\leq O_{q,d,\epsilon}(n^{d-1}).
\end{align}

\subsection{Hypergraph state ensembles} \label{subsection: hypergraph states}
Let $(V,E)$ be a hypergraph as defined in \Cref{subsection: hypergraph state intro}, with edges of size at most $d$. Each subhypergraph $a\in\bF_2^E$ corresponds to the state 
\begin{align}
    \ket{\psi_a}=2^{-|V|/2}\sum_{x\in\bF_2^V}(-1)^{\sum_{e\in E}a_e\prod_{v\in e}x_v}\ket{x}.
\end{align}
This is a GU ensemble generated by $G=A\rtimes K$ where $A\cong\bF_2^E$ and $K$ is the automorphism group of the hypergraph, and $\ket{\psi_0}=2^{-|V|/2}\sum_{x\in\bF_2^V} \ket{x}$. It admits a feature map $\Phi\colon \bF_2^V\to\bF_2^E$ given by 
\begin{align}
    \Phi(x)\coloneqq \left(\prod\nolimits_{v\in e}x_v\right)_{e\in E}
\end{align}
and the lifted dual walk can be written as
\begin{align}
    (Z_t)_e=\sum_{j=1}^t\prod_{v\in e}(X_j)_v,\qquad X_j\stackrel{\rm i.i.d.}{\sim}\Unif(\bF_2^V).
\end{align}
Its generalized Dobrushin coefficient can be fully expressed in terms of hypergraph data: 
\begin{align} \label{eq: hypergraph Dobrushin coefficient}
    \kgen=1-\sup_{s\in\mathbb N}\frac{1}{s2^{s|V|}}\min_{a\subset E}\sum_{z\subset E}\min\{R_s(z),R_s(z\triangle a)\}
\end{align}
where
\begin{align}
    R_s(z)\coloneqq\#\{(S_1,\ldots,S_s)\in(2^V)^s:E(S_1)\triangle\ldots\triangle E(S_s)=z\}
\end{align}
with $E(S)\coloneqq\{e\in E:e\subset S\}$ for any $S\subset V$. We also provide the following upper bound in terms of hyperedge data:

\begin{prop} \label{proposition: hypergraph}
    Given any ordering of $V$,
    \begin{align}
        t_{\min}(\epsilon)\leq 2^d\left(\log\left(\sum\nolimits_{v\in V}2^{|E_v|}\right)-\log\epsilon\right).
    \end{align}
    where $E_v$ is the set of hyperedges whose largest vertex in this ordering is $v$.
\end{prop}
\begin{proof}
    The feature map can be written in the triangular form as
    \begin{align}
        \Phi_v(x)=x_v\left(\prod\nolimits_{u\in e\backslash\{v\}}x_u\right)_{e\in E_v}
    \end{align}
    with $B_v\cong\bF_2^{|E_v|}$. The characters are $\chi_c(z)=(-1)^{\sum_{e\in E_v}c_ez_e}$ for $c\in B_v\backslash\{0\}$, so
    \begin{align}
        \chi_c(\Phi_v(X_{\leq v}))=(-1)^{X_vh_{v,c}(X_{<v})}
    \end{align}
    where
    \begin{align}
        h_{v,c}(x_{<v})=\sum_{e\in E_v}c_e\prod_{u\in e\backslash\{v\}}x_u.
    \end{align}
    By \Cref{lemma: polynomial anti-concentration}, 
    \begin{align}
        \lambda_v=\max_{c\neq0}\bE_{X_{<v}}|\bE_{X_v}\chi_c(\Phi_v(X_{\leq v}))|^2=\max_{c\neq 0}P(h_{v,c}(X_{<v})=0)\leq1-2^{1-d}.
    \end{align}
    By \Cref{thm: total variation upper bound for triangular feature map},
    \begin{align}
        d_{\TV}(\law(Z_t),\Unif)\leq\frac{1}{2}\sum_{v\in V}\sqrt{(2^{|E_v|}-1)(1-2^{1-d})^t},
    \end{align}
    therefore
    \begin{align}
        t_{\min}(\epsilon)\leq\frac{2\log(\sum_{v\in V}2^{|E_v|})-2\log\epsilon}{-\log(1-2^{1-d})}\leq 2^d\left(\log\left(\sum\nolimits_{v\in V}2^{|E_v|}\right)-\log\epsilon\right).
    \end{align}
\end{proof}

\subsubsection{Graph state ensemble}
We now specialize the result from the previous section to ordinary graph states.
Given a graph $(V,E)$, we identify the set of subgraphs with $A\coloneqq \bF_2^{|E|}$, and the graph state corresponding to a subgraph $a\in A$ can be written as
\begin{align}
    \ket{\Gamma_a}
  =\prod_{e\in E}\text{CZ}_e^{a_e}|+\rangle^{\otimes |V|}=2^{-\frac{|V|}{2}}\sum_{x\in\bF_2^V}
  (-1)^{\sum_{e=\{u,v\}\in E}a_ex_ux_v}\ket{x}.
\end{align}

\begin{cor} \label{cor: graph states}
    Fix an ordering of $V$, let $b_v\coloneqq|\{u\in N(v):u<v\}|$, then
    \begin{align}
        t_{\min}(\epsilon)\leq O\big(\log(\sum_{v\in V}2^{b_v})-\log\epsilon\big).
    \end{align}
\end{cor}
\begin{proof}
    Graphs are $2$-uniform hypergraphs, i.e., all edges consist of two vertices, so this is an immediate corollary of \Cref{proposition: hypergraph}.
\end{proof}

The following result for the sample complexity of learning hidden graph states is an immediate corollary of \Cref{cor: graph states} and improves \cite[Theorem 15]{MontanaroShao2022HiddenGraph}. Note that it concerns only sample complexity and does not reproduce their computational runtime guarantees.

\begin{cor}
    Let $G'$ be a graph of bounded degeneracy $D$, and $G$ be a subgraph of $G'$. Given access to copies of $\ket{G}$, there is a quantum algorithm that identifies $G$ using $O_\epsilon(D+\log n)$ copies.
\end{cor}
\begin{proof}
    Choose a degeneracy ordering so that every vertex has at most $D$ earlier neighbors, by \Cref{cor: graph states} we obtain $t_{\min}\leq O_\epsilon(D+\log n)$.
\end{proof}

\section{Generic non-GU ensembles} \label{section: generic non-GU ensembles}
In this section, we discuss how the sample complexity is related to a mixing time in the non-GU case. We first consider arbitrary uniform mixed state ensembles $\{\rho_i=S_iS_i^\dagger\}_{i=1}^{n}$ on $\bC^d$ and assume $F(\rho_i,\rho_j)<1$ for $i\neq j$. Define
\begin{align}
    X_t\coloneqq \begin{bmatrix}(S_i^\dagger S_j)^{\otimes t}\end{bmatrix}_{i,j=1}^{n}
\end{align}
and
\begin{align}
    t_{\min}(\epsilon)\coloneqq \min \left\{t:\frac{1}{n}\|X_t-\Delta_t(X_t) \|_1\leq\epsilon\right\}
\end{align}
where $\Delta_t(X_t)$ is the block pinching of $X_t$, i.e., 
\begin{align} 
    \Delta_t(X_t)=[\delta_{ij}(S_i^\dagger S_j)^{\otimes t}]_{i,j=1}^{n}.
    \label{eq:block-pinching}
\end{align}
Note that $t_{\min}$ is independent of the choice of $S_i$ and thus well-defined given the ensemble $\{\rho_i\}$.

We now define a quantum channel 
\begin{align}
    \cN_t\colon M_n\otimes M_d^{\otimes(t-1)}&\to M_n\otimes M_d^{\otimes t}\\
    [A_{ij}]&\mapsto[A_{ij}\otimes S_i^\dagger S_j]
\end{align}
with Kraus operators given by $L_{t,a}\coloneqq\sum_i E_{ii}\otimes\one_d^{\otimes(t-1)}\otimes S_i^\dagger\ket{a}$. With $\Phi_t\coloneqq\cN_t\circ\ldots\circ\cN_1,$ we have
\begin{align}
    \frac{1}{n}X_t=\Phi_t(\tau_n),
\end{align}
where $\tau_n\coloneqq \frac{1}{n}\mathbb J$ is the normalization of the all-ones matrix $\mathbb J$. This is a quantum inhomogeneous Markov chain, and therefore 
\begin{align}
    t_{\min}(\epsilon)=\min\{t:\|\Phi_t(\tau_n)-\Phi_t(\omega_n)\|_1\leq\epsilon\},
\end{align}
where $\omega_n$ is the completely mixed state on $\bC^n$.
In the terminology of \cite{George2026QuantumDoeblinCoefficients}, $t_{\min}$ is the \emph{quantum weakly mixing time} with fixed initial state $\tau_n$ and reference state $\omega_n$.

The goal of this section is to establish a connection between the state discrimination sample complexity $k_{\min}$ and the mixing time (with fixed initial state) $t_{\min}$ for arbitrary uniform ensembles. We first prove the following lemma for bounding the success probability of the PGM using Gram matrix information:

\begin{lem} \label{lemma: block Gram matrix lemma}
    Given $\rho_i=S_iS_i^\dagger$, denote the Gram matrix by $G\coloneqq X_1=\begin{bmatrix}S_i^\dagger S_j\end{bmatrix}_{i,j=1}^{n}$, and set $\Sigma\coloneqq \sum_i\rho_i$.
    Let $\Delta\equiv\Delta_1$ be the block pinching map defined in \eqref{eq:block-pinching}. Then the following holds:
    \begin{enumerate}
    \item $1-p_{\mathrm{PGM}}=\frac{1}{n}\|\sqrt{G}-\Delta(\sqrt{G})\|_2^2\leq\frac{1}{n}\|\sqrt{G}-\sqrt{\Delta(G)}\|_2^2$.\vspace{.25em}
    \item $\|G-\Delta(G)\|_1^2\leq 8n^2(1-p_{\mathrm{PGM}})$.
    \end{enumerate}
\end{lem}
\begin{proof}
    (1) Denote PGM operators by $M_i=\Sigma^{-\frac{1}{2}}\rho_i\Sigma^{-\frac{1}{2}}$, then
    \begin{align}
        \sum_{i,j}\|(\sqrt{G})_{ij}\|_2^2=\|\sqrt{G}\|_2^2=\tr(G)=n;\qquad \tr(M_i\rho_i)=\|S_i^\dagger\Sigma^{-\frac{1}{2}}S_i\|_2^2=\|(\sqrt{G})_{ii}\|_2^2.
    \end{align}
    Therefore
    \begin{align}
        1-p_{\mathrm{PGM}}=1-\frac{1}{n}\sum_i\|(\sqrt{G})_{ii}\|_2^2=\frac{1}{n}\sum_{i\neq j}\|(\sqrt{G})_{ij}\|_2^2=\frac{1}{n}\sum_{i,j}\|Z_{ij}\|_2^2=\frac{1}{n}\|Z\|_2^2
    \end{align}
    where $Z=\sqrt{G}-\Delta(\sqrt{G})$, $Z_{ii}=0$ and $Z_{ij}=(\sqrt{G})_{ij}$ for $i\neq j$. Finally, since $\sqrt{\Delta(G)}$ is block-diagonal,
    \begin{align}
        \|\sqrt{G}-\sqrt{\Delta(G)}\|_2^2 &=\sum_{i\neq j}\|(\sqrt{G})_{ij}\|_2^2+\sum_i\|(\sqrt{G})_{ii}-(\sqrt{\Delta(G)})_{ii}\|_2^2\\
        &\geq\sum_{i\neq j}\|(\sqrt{G})_{ij}\|_2^2\\
        &=\|Z\|_2^2.
    \end{align}
    (2) Since $G_{ii}=(\sqrt{G})_{ii}^2+\sum_{j\neq i}(\sqrt{G})_{ij}(\sqrt{G})_{ij}^\dagger\geq(\sqrt{G})_{ii}^2$,
    \begin{align}
        \tr[G_{ii}-(\sqrt{G})_{ii}^2]=\tr[2(\sqrt{G_{ii}}-(\sqrt{G})_{ii})(\sqrt{G})_{ii}+(\sqrt{G_{ii}}-(\sqrt{G})_{ii})^2]\geq\|\sqrt{G_{ii}}-(\sqrt{G})_{ii}\|_2^2.
    \end{align}
    Moreover, since $(\sqrt{G}-\sqrt{\Delta(G)})_{ij}$ is equal to $(\sqrt{G})_{ij}$ if $ i\neq j$, and equal to $(\sqrt{G})_{ii}-\sqrt{G_{ii}}$ if $i=j$, 
    \begin{align}
        \|\sqrt{G}-\sqrt{\Delta(G)}\|_2^2 &= \sum_{i\neq j}\|(\sqrt{G})_{ij}\|_2^2+\sum_i\|(\sqrt{G})_{ii}-\sqrt{G_{ii}}\|_2^2\\
        &\leq\sum_{i\neq j}\|(\sqrt{G})_{ij}\|_2^2+\sum_i\tr(G_{ii}-(\sqrt{G})_{ii}^2)\\
        &=2\sum_{i\neq j}\|(\sqrt{G})_{ij}\|_2^2\\
        &=2n(1-p_{\mathrm{PGM}}).
    \end{align}
    Finally, we have $G-\Delta(G)=(\sqrt{G}-\sqrt{\Delta(G)})\sqrt{G}+\sqrt{\Delta(G)}(\sqrt{G}-\sqrt{\Delta(G)})$, and hence applying the triangle inequality and Hölder's inequality gives
    \begin{align}
        \|G-\Delta(G)\|_1 &\leq \left\|\sqrt{G}-\sqrt{\Delta(G)} \right\|_2 \left(\|\sqrt{G}\|_2+\|\sqrt{\Delta(G)}\|_2\right)\\
        &=2\sqrt{n}\left\|\sqrt{G}-\sqrt{\Delta(G)}\right\|_2\\
        &\leq2\sqrt{2}n\sqrt{1-p_{\mathrm{PGM}}},
    \end{align}
    concluding the proof.
\end{proof}

Now we establish a sandwiched bound for the sample complexity $k_{\min}$ by the mixing time $t_{\min}$:

\MixedStateSandwich*
\begin{proof}
    For the upper bound, by \Cref{lemma: block Gram matrix lemma}(1) and Powers–Størmer inequality \cite{KittanehKosaki1987InequalitiesSchattenPNormV,zhou2026majorization},
    \begin{align}
        1-p_t^*\leq1-p_{\mathrm{PGM}}\leq\frac{1}{n}\|\sqrt{X_t}-\sqrt{\Delta(X_t)}\|_2^2\leq\frac{1}{n}\|X_t-\Delta(X_t)\|_1,
    \end{align}
    therefore $t_{\min}(\epsilon)\geq k_{\min}(\epsilon)$. For the lower bound, by \Cref{lemma: block Gram matrix lemma}(2) and $p_{\mathrm{PGM}}\geq(p^*)^2$,
    \begin{align}
        2(1-p_t^*)\geq 1-(p_t^*)^2\geq1-p_{\mathrm{PGM}}\geq\frac{1}{8n^2}\|X_t-\Delta(X_t)\|_1^2,
    \end{align}
    and therefore $t_{\min}(\epsilon)\leq k_{\min}(\frac{\epsilon^2}{16})$.
\end{proof}

\begin{rem}
    This sandwiched bound generalizes \cref{eq: sandwiched bound for GU sample complexity} to arbitrary uniform mixed state ensembles. However, without an additional structure (such as GU) there is no guarantee that $k_{\min}=\Theta_\epsilon(t_{\min})$, even in the pure state case.
\end{rem}

\subsection{Tight estimate for worst-case discrimination sample complexity by Dobrushin-type coefficient}

To follow this section, it will be useful to recall the discussion of (single-step and multi-step) quantum Dobrushin coefficients in \Cref{subsubsection: Quantum Dobrushin coefficient}.
Here, we make use of the notion of \emph{generalized Dobrushin ergodicity coefficient}, defined by \cite{mukhamedov2020generalized}
\begin{align} \label{eq: generalized Dobrushin ergodicity coefficient}
    \tilde{\kappa}(T)\coloneqq\sup_{\substack{X=X^\dagger\neq0\\\Delta(X)=0}}\frac{\|T(X)\|_1}{\|X\|_1}.
\end{align}
We define the multi-step coefficient by
\begin{align} \label{eq: stabilized multistep coefficient}
    \tkgen\coloneqq1-\sup_{s\geq1}\frac{1-\eta(\Phi_s)}{s}
\end{align}
where $\eta_s$ is the stabilized Dobrushin ergodicity coefficient defined by
\begin{align}
    \eta(T)\coloneqq\sup_R\tilde{\kappa}(T\otimes\id_R).
\end{align}
This stabilization is necessary in general because we need to characterize the contraction rate of
\begin{align}
    \Phi_{s+t}=(\Phi_s\otimes\id_{\cH^{\otimes t}})\circ\Phi_t.
\end{align}
The sandwiched bound \eqref{eq: sandwiched bound for sample complexity} is not guaranteed to be tight, and the multi-step coefficient does not always give a tight estimate either; see \cref{eq: min-error simple counterexample} for a simple example. What it unconditionally describes turns out to be the worst-state (a.k.a. minimax) discrimination sample complexity \cite{DAriano2005Minimax,Montanaro2019PrettySimpleBounds}:

\WorstCaseDiscrimination*
\begin{proof}
    We claim that
    \begin{align} \label{eq: worst-case sandwiched bound}
        1-\pwc(t)\leq 2\eta(\Phi_t)\leq 4\sqrt{1-\pwc(t)}.
    \end{align}
    Assuming this claim, given any $t$ with $\eta(\Phi_t)<1$, we have $1-\pwc(rt)\leq2\eta_{rt}\leq2\eta(\Phi_t)^r$ for all $r$ by the contraction property. Taking $r=\lceil\frac{\log(\frac{2}{\epsilon})}{1-\eta(\Phi_t)}\rceil$ and an infimum over $t$ shows that
    \begin{align}
        \kwc(\epsilon)\leq\inf_t\{rt\}\leq C_\epsilon\inf_t\frac{t}{1-\eta(\Phi_t)}=O_\epsilon\left(\frac{1}{1-\tkgen}\right).
    \end{align}
    On the other hand, since $\kwc(\epsilon)=\Theta_\epsilon(\kwc(\epsilon^2))$ by \Cref{proposition: worst-case sample complexity equivalence}, we have
    \begin{align}
        1-\tkgen\geq\frac{1-\eta(\Phi_{c_\epsilon\kwc(\epsilon)})}{c_\epsilon\kwc(\epsilon)}\geq\frac{1-2\epsilon}{c_\epsilon\kwc(\epsilon)}.
    \end{align}
    It remains to prove \cref{eq: worst-case sandwiched bound}. 
    Since $L_{\vec a}=\sum_iE_{ii}\otimes S_i^\dagger\ket{a_1}\otimes\ldots\otimes S_i^\dagger\ket{a_t}$ form a set of Kraus operators of $\Phi_t$, 
    the complementary channel of $\Phi_t$ can be written as $\Phi_t^c(Y)=\sum_{\vec{a},\vec{b}}\tr(L_{\vec a}YL_{\vec b}^\dagger)\ket{\vec a}\bra{\vec b}=\sum_i\bra{i}Y\ket{i}(\rho_i^{\otimes t})^T$. Therefore we can apply \cite[eq. (116) and (119)]{beny2011approximate} to $\Phi_t^c$ and obtain
    \begin{align} \label{eq: Beny-Oreshkov}
        \sqrt{\pwc(t)}=\max_{\mc R}F(\mc R\circ\Phi_t^c,\Delta)=\max_{\mc R'}F(\Phi_t,\mc R'\circ\Delta)
    \end{align}
    where $F$ denotes the worst-case entanglement fidelity and the second equality is \cite[Theorem 1]{beny2010general}. By the contraction property of $\eta$,
    \begin{align}
        \|\Phi_t-\Phi_t\circ\Delta\|_\diamond\leq\eta(\Phi_t)\sup_{\|X\|_1\leq1}\|X-(\Delta\otimes\id_R)(X)\|_1\leq 2\eta(\Phi_t),
    \end{align}
    therefore, by the Fuchs-van de Graaf inequality \cite[eq. (58)]{BelavkinDArianoRaginsky2005}, 
    \begin{align}
        1-\pwc(t)\leq2(1-\sqrt{\pwc(t)})\leq2(1-F(\Phi_t,\Phi_t\circ\Delta))\leq\|\Phi_t-\Phi_t\circ\Delta\|_\diamond\leq2\eta(\Phi_t).
    \end{align}
    On the other hand, for every $\mc R'$ maximizing \cref{eq: Beny-Oreshkov}, every reference system $R$ and every admissible $X$ such that $(\Delta\otimes\id_R)(X)=0$, we have $(\Phi_t\otimes\id_R)(X)=((\Phi_t-\mc R'\circ\Delta)\otimes\id_R)(X)$. Taking a supremum and applying the Fuchs–van de Graaf inequality again we obtain
    \begin{align}
        \eta(\Phi_t)\leq\|\Phi_t-\mc R'\circ\Delta\|_\diamond\leq2\sqrt{1-F(\Phi_t,\mc R'\circ\Delta)^2}=2\sqrt{1-\pwc(t)}.
    \end{align}
    The chain of inequalities is simply \cref{eq: sandwiched bound for sample complexity} and 
    \begin{align}
        t_{\min}(\epsilon)\leq k_{\min}(\frac{\epsilon^2}{16})\leq \kwc(\frac{\epsilon^2}{16})=\Theta_\epsilon (\kwc(\epsilon)),
    \end{align}
    finishing the proof.
\end{proof}

Eq. \eqref{eq: sandwiched bound for sample complexity} is tight when $k_{\min}(\epsilon)=\Theta_\epsilon(k_{\min}(\frac{\epsilon^2}{16}))$. We leave the structural classification of this condition to future research.

\subsection{Pure state ensembles}
For pure state ensemble $\lbrace |\psi_i\rangle\langle\psi_i|\rbrace$ we can choose $S_i=\ket{\psi_i}$.
Then $(X_t)_{ij}=\braket{\psi_i|\psi_j}^t$ and
\begin{align} \label{eq: mixing time for pure states}
    t_{\min}(\epsilon)=\min\{t:\|\Phi_X^t(\tau_n)-\omega_n\|_1\leq\epsilon\}
\end{align}
reduces to a quantum homogeneous mixing time of the Hadamard channel $\Phi_X$. An important simplification of the multi-step coefficient \eqref{eq: stabilized multistep coefficient} in the pure state case is that the stabilization is not needed:

\begin{prop}
    For pure states, the stabilized and unstabilized coefficients bound each other as
    \begin{align}
        \tilde{\kappa}(\Phi_X^t)\leq\eta(\Phi_t)\leq2\tilde{\kappa}(\Phi_X^t).
    \end{align}
\end{prop}
\begin{proof}
    For any unit vector $\ket{\psi}=\sum_i a_i\ket{i}\otimes\ket{u_i}$, denote $\ket{a}=\sum_i a_i\ket{i}$ and consider the isometry $V\colon\ket{i}\mapsto\ket{i}\otimes\ket{u_i}$.
    Then,
    \begin{align}
        \|((\Phi_X^t-\Delta)\otimes\id)(\ket{\psi}\bra{\psi})\|_1&=\|\sum_{i\neq j}a_ia_jX_{ij}^tE_{ij}\otimes\ket{u_i}\bra{u_j}\|_1\\
        &=\|V\Phi_X^t(\ket{a}\bra{a}-\Delta(\ket{a}\bra{a}))V^\dagger\|_1\\
        &\leq2\tilde{\kappa}(\Phi_X^t).
    \end{align} 
    For any Hermitian $X=\sum_k\lambda_k\ket{\psi_k}\bra{\psi_k}$ such that $(\Delta\otimes\id)(X)=0$, 
    \begin{align}
        \|(\Phi_X^t\otimes\id)(X)\|_1=\|((\Phi_X^t-\Delta)\otimes\id)(X)\|_1\leq2\tilde{\kappa}(\Phi_X^t)\sum_k|\lambda_k|=2\tilde{\kappa}(\Phi_X^t)\|X\|_1,
    \end{align}
    therefore $\eta(\Phi_t)\leq2\tilde{\kappa}(\Phi_X^t)$. 
    The lower bound is true for all mixed states, in particular for pure states.
\end{proof}

In the GU pure state case, zero-trace is equivalent to zero-diagonal for any matrix in the Hecke algebra because it always has identical diagonal elements.
Therefore, when restricted to the Hecke algebra, the generalized Dobrushin ergodicity coefficient defined in \cref{eq: generalized Dobrushin ergodicity coefficient} coincides with the standard quantum Dobrushin coefficient defined in \cref{eq: defn of quantum Dobrushin coefficient}, which matches the coefficient in \Cref{theorem: GU Dobrushin estimate}.

\subsubsection{A strengthened data processing inequality upper bound}
As a parallel approach, we consider a strengthened data processing inequality (SDPI) for relative entropy with respect to a conditional expectation \cite{gao2022complete,hirche2022contraction}, which gives an immediate upper bound for estimating $t_{\min}$.

\begin{lem} \label{lemma: SDPI}
    Let $\cE$ be a trace preserving conditional expectation onto a von Neumann subalgebra $\cN$. For any quantum channel $\Phi$ such that $\cE\circ\Phi=\Phi\circ\cE=\cE$, let $D_\cE(\rho)\coloneqq D(\rho\|\cE(\rho))$ and define
    \begin{align}
        \eta(\Phi,\cE)\coloneqq\sup_{\rho:D_\cE(\rho)>0}\frac{D_\cE(\Phi(\rho))}{D_\cE(\rho)}.
    \end{align}
    Then
    \begin{align}
        \|\Phi^t(\rho)-\cE(\rho)\|_1\leq2\sqrt{\eta(\Phi,\cE)^tD_\cE(\rho)}.
    \end{align}
\end{lem}
\begin{proof}
    By definition of $\eta(\Phi,\cE)$,
    \begin{align}
        D_\cE(\Phi^t(\rho))\leq\eta(\Phi,\cE)D_\cE(\Phi^{t-1}(\rho))\leq\ldots\leq\eta(\Phi,\cE)^tD_\cE(\rho).
    \end{align}
    By quantum Pinsker's inequality \cite{HiaiOhyaTsukada1981} and $\cE\circ\Phi=\cE$,
    \begin{align}
        \|\Phi^t(\rho)-\cE(\rho)\|_1\leq2\sqrt{D_\cE(\Phi^t(\rho))}\leq2\sqrt{\eta(\Phi,\cE)^tD_\cE(\rho)}.
    \end{align}
\end{proof}

In particular,
\begin{align}
    \|\Phi_X^t(\tau_n)-\omega_n\|_1 &\leq2\sqrt{\eta^tD_\Delta(\tau_n)}=2\sqrt{\eta^t\log n}
    \intertext{where}
    \eta &=\eta(\Phi_X,\Delta)=\sup_{\rho:D(\rho\|\Delta(\rho))>0}\frac{D(\Phi_X(\rho)\|\Delta\circ\Phi_X(\rho))}{D(\rho\|\Delta(\rho))}.
\end{align}
Therefore we have proved:
\begin{prop} \label{prop: SDPI}
    When $n>1$ and $\eta\in(0,1)$,
    \begin{align} \label{eq: SDPI}
    t_{\min}(\epsilon)\leq O\left(\frac{\log\log n-2\log\epsilon}{-\log\eta}\right).
\end{align}
\end{prop}

In particular, when $\eta$ is bounded away from 1 this gives an $O_\epsilon(\log\log n)$ upper bound, which is a nontrivial improvement of the standard $O_\epsilon(\log n)$ bound in \Cref{theorem: Montanaro upper bound of k_min}. We leave further exploration in the SDPI method to future research.

\section{Conclusion}
In this work we develop a mixing-time method for estimating the sample complexity of quantum state discrimination. For arbitrary mixed-state ensembles with uniform priors, we establish two-sided bounds on the minimum-error discrimination sample complexity in terms of a quantum weakly mixing time. We also obtain a tight characterization of the worst-case discrimination sample complexity in terms of a generalized Dobrushin coefficient.

For geometrically uniform pure-state ensembles, the relevant mixing time is the quantum homogeneous mixing time of the Hadamard channel induced by the Gram matrix, which tightly characterizes the minimum-error discrimination sample complexity. When the generating group and the stabilizer subgroup of the generator state form a Gelfand pair, this quantum mixing problem further reduces to a classical one. In this setting, the generalized Dobrushin coefficient can be expressed entirely in terms of representation-theoretic quantities of the commutative Hecke algebra.

We apply this framework to recover sample-complexity estimates for quantum coupon collector states and to analyze a class of phase-state ensembles for which the mixing problem reduces to that of an abelian random walk. Assuming a triangular structure of the feature map defining the phase states, we establish a general upper bound on the mixing time, with applications to learning degree-$d$ phase states, generalized Boolean phase states, and hypergraph states.

Future directions include characterizing the conditions under which mixing-time bounds tightly estimate the sample complexity for arbitrary ensembles, investigating cutoff phenomena in quantum state discrimination, and developing further bounds based on strengthened data processing inequalities.

\medskip
\paragraph*{Acknowledgments} The authors acknowledge useful discussions with Srinivasan Arunachalam, Hari Krovi, Kristan Temme, Mark Wilde and Pawel Wocjan. 
JZ acknowledges the hospitality of Srinivasan Arunachalam and IBM Research Almaden during an externship in March 2025, which inspired some of the ideas for this work.
This research was supported by a grant through the IBM-Illinois Discovery Accelerator Institute, as well as National Science Foundation Grant No.~2426103.

\paragraph*{Disclosure of AI usage} The main idea and overall structure of this article were established in late
2025 by the authors. The literature search, verification of calculations, and refinement of the writing were assisted by GPT-5 and subsequent models. The authors have reviewed and verified the entire article and assume full
responsibility for its accuracy and correctness.

\printbibliography[heading=bibintoc]
\end{document}